%% file: Camera_Ready_Redist_2024.tex
\documentclass[letterpaper]{article} 
\usepackage{aaai24}  
\usepackage{times}  
\usepackage{helvet}  
\usepackage{courier}  
\usepackage[hyphens]{url}  
\usepackage{graphicx} 
\usepackage{natbib}  
\usepackage{caption} 
\usepackage{algorithm}
\usepackage{algorithmic}

\usepackage{newfloat}
\usepackage{listings}
\DeclareCaptionStyle{ruled}{labelfont=normalfont,labelsep=colon,strut=off} 
\floatstyle{ruled}
\newfloat{listing}{tb}{lst}{}
\floatname{listing}{Listing}
\title{Implications of Distance over Redistricting Maps: Central and Outlier Maps\footnote{Readers are encouraged to see the arxiv version with more results: https://arxiv.org/abs/2203.00872}}
\author{
    Seyed A. Esmaeili \textsuperscript{\rm 1}, 
    Darshan Chakrabarti\textsuperscript{\rm 2}, Hayley Grape\textsuperscript{\rm 3}, Brian Brubach\textsuperscript{\rm 3}
}
\affiliations{
    \textsuperscript{\rm 1}University of Chicago Data Science Institute\\
    \textsuperscript{\rm 2}Columbia University\\
    \textsuperscript{\rm 3}Wellesley College\\
    esmaeili@uchicago.edu, darshan.chakrabarti@columbia.edu, hg3@wellesley.edu, bb100@wellesley.edu 

}

\usepackage{thm-restate}
\usepackage{amsmath,amsfonts,amssymb,amsthm,  mathrsfs,mathtools}
\usepackage{booktabs}
\usepackage{multirow}
\usepackage{subfig}
\usepackage{xcolor}         

\input{macros}

\begin{document}

\maketitle

\begin{abstract}
In representative democracy, a redistricting map is chosen to partition an electorate into districts which each elects a representative. A valid redistricting map must satisfy a collection of constraints such as being compact, contiguous, and of almost-equal population. However, these constraints are loose enough to enable an enormous ensemble of valid redistricting maps. This enables a partisan legislature to gerrymander by choosing a map which unfairly favors it. In this paper, we introduce an interpretable and tractable distance measure over redistricting maps which does not use election results and study its implications over the ensemble of redistricting maps. Specifically, we define a central map which may be considered "most typical" and give a rigorous justification for it by showing that it mirrors the Kemeny ranking in a scenario where we have a committee voting over a collection of redistricting maps to be drawn. We include running time and sample complexity analysis for our algorithms, including some negative results which hold using any algorithm. We further study outlier detection based on this distance measure and show that our framework can detect some gerrymandered maps. More precisely, we show some maps that are widely considered to be gerrymandered that lie very far away from our central maps in comparison to a large ensemble of valid redistricting maps. Since our distance measure does not rely on election results, this gives a significant advantage in gerrymandering detection which is lacking in all previous methods.
\end{abstract}

\input{introduction.tex}

\input{relatedWork.tex}

\input{ProblemSetup.tex}
\input{Algorithms.tex}
\input{Experiments.tex}

\input{Conclusion.tex}

\bibliography{refs}

\appendix
\input{Appendix}

\end{document}

%% file: macros.tex
\DeclareMathOperator*{\argmin}{arg\,min}

\DeclareMathOperator{\E}{\mathbb{E}}

\newtheorem{theorem}{Theorem}[section]

\newtheorem{condition}{Condition}[section]

\newtheorem{lemma}{Lemma}

\newtheorem{remark}{Remark}

\DeclareMathOperator{\map}{\mathnormal{M}}
\DeclareMathOperator{\MapsDist}{\mathcal{M}}

\DeclareMathOperator{\Maps}{\mathcal{M}_\mathnormal{T}}
\DeclareMathOperator{\Dist}{\mathcal{D}}

\newcommand{\norm}[1]{\left\lVert#1\right\rVert}

\DeclareMathOperator{\dtheta}{\mathnormal{d_{\Theta}}}
\DeclareMathOperator{\dthetasq}{\mathnormal{d_{2,\Theta}}}

\DeclareMathOperator{\ac}{\mathnormal{\bar{A}_c}}
\DeclareMathOperator{\acpop}{\mathnormal{A_c}}

\DeclareMathOperator{\amed}{\mathnormal{\bar{A}^*}}
\DeclareMathOperator{\amedpop}{\mathnormal{A^*}}

\DeclareMathOperator{\acT}{\mathnormal{\bar{A}^{\Theta}_c}}

\DeclareMathOperator{\iid}{\textbf{iid}}

\DeclareMathOperator{\recom}{\textbf{ReCom}}

\newcommand{\mh}[1]{\mathnormal{\hat{M}_{#1}}}
\DeclareMathOperator{\RE}{\textbf{RE}}

\DeclareMathOperator{\ag}{\mathnormal{A_G}}

%% file: introduction.tex
\section{Introduction}
Redistricting is the process of dividing an electorate into 
districts which each elect a representative. In the United States, this process is used for both federal and state-level representation, and we will use the U.S. House of Representatives as a running example. 
Subject to both state and federal law, the division of states into congressional districts is not arbitrary and must satisfy a collection of properties such as contiguity and near-equal population. Despite these regulations, 
redistricting is vulnerable to strategic manipulation in the form of gerrymandering. The body in charge 
can easily create a map within the legal constraints that leads to election results which favor a particular outcome (e.g., more representatives elected from one political party in the case of \emph{partisan gerrymandering}). In addition, the ability to draw gerrymandered districts has improved greatly with the aid of computers since the historic salamander-shaped district approved by Massachusetts Governor Elbridge Gerry in 1812. For example, assuming voting consistent with the 2016 election, the state of North Carolina with 13 representatives can be redistricted to elect either 3 Democrats and 10 republicans or 10 Democrats and 3 Republicans. 

However, despite this obvious threat to functioning democracy, partisan gerrymandering has often eluded regulation partly because it has been difficult to measure. In response, a recent line of impactful research introduced sampling techniques to 
randomly generate a large collection of redistricting maps\footnote{These are not the truly uniform random samples from the immense and ill-defined space of all possible maps that we ideally want, but they are generally treated as such in courts.} \cite{Chikina2017,deford2019recombination,herschlag2020quantifying} and calculate statistics such as a histogram of the number of seats won by each party using this collection. 
This can show that a proposed or enacted map is an outlier in terms of its election outcome with respect to the sample. 
In fact, these techniques were used as a key argument in a recent U.S. Supreme Court case on partisan gerrymandering~\cite{CaseRucho2019} and have supported successful efforts to change redistricting maps in state supreme court cases~\cite{LWVPenn2018}.
More importantly for the present work, at least two states, Michigan and Wisconsin, have used such a sampling tool~\cite{deford2019recombination} in the recent redistricting process in response to the 2020 U.S. Census ~\cite{TuftsRedistricting2021}. 
Recent papers have even extended these methods to locate which regions in a state are most unfairly impacted by a given map~\cite{lin2022auditing,ko2022all}.

Despite this progress, all of these methods rely on election outcomes to detect any possible gerrymandering. 
This is a problem in instances where citizens, courts, and/or legislative bodies request methods which do not take partisanship into account. 
However, there has not been an effective method for detecting gerrymandering by identifying outliers according to a non-partisan metric.  
\citet{abrishami2020geometry} make progress toward this goal, but has limitations such a using a small number of samples and not having a clear outlier score (See Section \ref{sec:rw} and Appendix \ref{app:comparison_Abri} for more details). 
In this paper, we introduce a framework that resolves these issues and demonstrates the first effective method for detecting gerrymandering based on a distance (dissimilarity) measure over redistricting maps. 

Furthermore, while progress has been made on the problem of detecting/labeling gerrymandering through the use of these sampling techniques, the question of drawing a redistricting map in a way that is ``fair'' and resists strategic manipulation remains largely unclear. We survey some existing proposals to automate redistricting in more detail in Section \ref{sec:rw}, but none of them have been adopted in practice thus far. 
Indeed, the issue of finding the ``ideal'' redistricting map is elusive and one of the main problems in redistricting and gerrymandering. We make progress in this direction, by introducing a novel method which selects the most ``central'' map. More specifically, given a collection of maps which are voted on by committee members we select the map with the minimum vote-weighted distance from the collection.  





\subsection{Our Contributions}
Our paper introduces a number of contributions: 
\begin{enumerate}
    \item \textbf{Distance over Redistricting Maps}: We introduce a tractable family of distance measures over redistricting maps which have a simple edit distance interpretation. This family of distance measures can be adjusted to accommodate considerations such as a population or path length between voting blocks (Subsection \ref{sec:distance_def}).
    
    \item \textbf{Medoid Map}: We introduce the \emph{medoid map} and show that it mirrors the Kemeny ranking in a setting where committee members vote over a collection of maps. We further characterize the complexity of finding this map and introduce heuristics for finding it (Sections \ref{sec:justification}, \ref{sec:algsandtheory}, \ref{sec:experiments}).
    
    \item \textbf{Centroid Map}: 
    We introduce the \emph{centroid map} which is not a valid map, but has interesting properties and implications. We provide algorithms for finding this map and characterize its sample complexity (Sections \ref{sec:algsandtheory},\ref{sec:experiments}).    
    
    \item \textbf{Gerrymandering Detection and Empirical Validation}: We show empirically that our framework can be used to detect gerrymandering in some instances. Further, we carry out extensive experimental validation of our pipeline where we ensure convergence and reproducibility by repeating the same experiment using different seeds. Remarkably, we reach the same conclusion across all runs: well-known gerrymandered maps of North Carolina and Pennsylvania have a distance in the 99th percentile away from the centroid map in comparison to an ensemble (Section \ref{sec:experiments} and Appendix \ref{app:exps}).   
\end{enumerate}
Furthermore, in Appendix \ref{app:gerry_detect} we include more details about gerrymandering detection such as the interpretation of having a large distance from the centroid map and a more detailed comparison to the work of \cite{abrishami2020geometry}. Finally, we note that our framework can adapt to various considerations such the Voting Rights Act (VRA) \cite{bickel1966voting} and other state-specific redistricting rules by simply modifying the sampling method to adapt to these considerations. Due to the space limits, we delay all proofs to Appendix \ref{app:math}.

%% file: relatedWork.tex
\section{Related Work}\label{sec:rw}
Less than a decade ago, several early works ushered in the current era of Markov Chain Monte Carlo (MCMC) sampling techniques for gerrymandering detection~\cite{mattingly2014redistricting,wu2015impartial,fifield2015new}. Followup work has both refined these techniques and further analyzed their ability to approximate the target distribution. Authors of these works have been involved in court cases in Pennsylvania~\cite{Chikina2017} and North Carolina~\cite{herschlag2020quantifying} with sampling approaches being used to demonstrate that existing maps were outliers as evidence of partisan gerrymandering. 
One of the most recent works in this area introduces the $\recom$ tool~\cite{deford2019recombination} which was used by the Wisconsin People’s Maps Commission and the Michigan Independent Citizens Redistricting Commission in the current redistricting cycle following the 2020 U.S. census~\cite{TuftsRedistricting2021}. More recently, the works of \cite{lin2022auditing,ko2022all} have made an interesting extension of the previous methods by identifying the voting blocks unfairly impacted in a gerrymandered map.  Generally, these techniques have primarily been used to analyze and sometimes reject existing maps rather than draw new maps. However, we may view them as narrowing the search space of maps that can be drawn. Furthermore, it has been shown that even the regulation of gerrymandering via outlier detection is subject to strategic manipulation~\cite{brubach2020meddling}.

On the automated redistricting side, many map drawing algorithms have favored optimization approaches and in particular, optimizing some notion of compactness while avoiding explicit use of partisan information. Approaches emphasizing compactness include balanced power diagrams~\cite{cohen2018balanced}, a $k$-median-based objective~\cite{Atlas538}, and minimizing the number of cut edges ~\cite{hettle2021balanced}. Some works include partisan information for the sake of creating competitive districts (districts with narrow margins between the two main parties). The PEAR tool~\cite{Liu2016} balances nonpartisan criteria like compactness (defined by the Polsby-Popper score~\cite{polsby1991third}) with other criteria such as competitiveness and uses an evolutionary algorithm with some similarity to the random walks taken by MCMC sampling approaches. Other works go further in the explicitly partisan direction such as the game theoretic approach of \citet{Pegden2017cake} which seeks a map that is fair to two parties. Finally, there are methods which prefer simplicity such as the Splitline~\cite{RangeVotingSplitline} algorithm which iteratively splits a state until the desired number of districts is reached.


In all of these approaches, the aim is to automate redistricting, but it is difficult to determine whether the choices made are the ``right'' or ``fairest'' decisions. The question of whether optimizing properties such as compactness while ignoring partisan factors could result in partisan bias is a concern. \citet{Cho2019technology} notes a comment by Justice Scalia suggesting that such a process could be biased against Democratic voters clustered in cities in Vieth v. Jubelirer~\cite{CaseVieth2004}. For those that do take partisan bias into account, there are questions of whether purposely drawing competitive districts or giving a fair allocation to two parties are really beneficial to voters.

Finally, the work of \citet{abrishami2020geometry} introduces a distance measure over redistricting maps. However, our distance is easy to compute and does not require solving a linear program. Further, our focus is on the implications of having a distance measure, i.e. the medoid and centroid maps that will be introduced. Moreover, unlike \citet{abrishami2020geometry} we can detect gerrymandered maps rigorously by specifying where they lie on a distance histogram without using an embedding method and using $200{,}000$ samples instead of only $100$. Finally, we give a clear outlier score for a given map (its percentile distance from the central map) whereas \cite{abrishami2020geometry} cannot do that. Therefore, it is difficult to see how their methods would be applied in a real practical setting. See Appendix \ref{app:comparison_Abri} for a more detailed discussion.

%% file: ProblemSetup.tex
\section{Problem Set-Up} \label{sec:problemsetup}
\begin{figure*}[t!]
  \centering
  \includegraphics[scale=0.33]{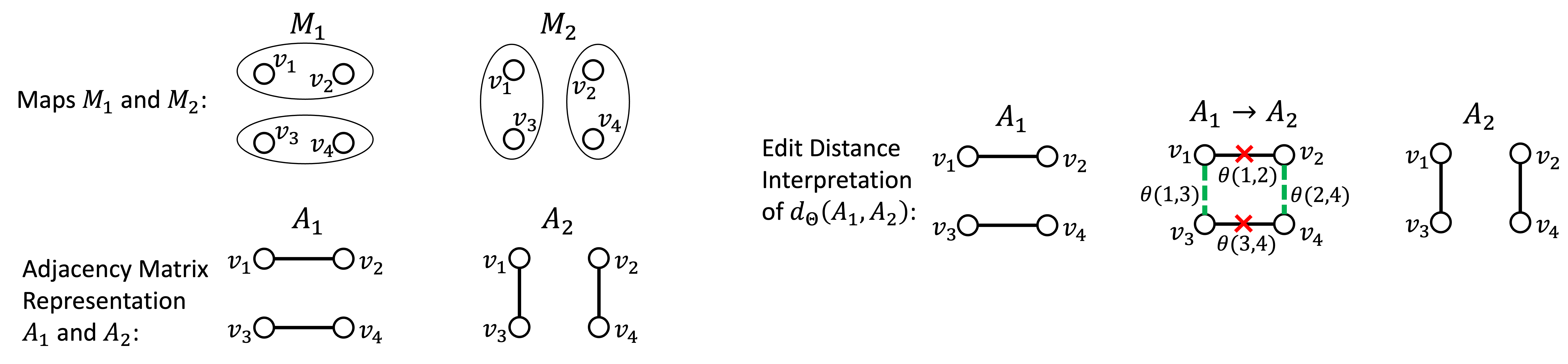}
  \caption{We are given a hypothetical state consisting of 4 vertices $V=\{v_1,v_2,v_3,v_4\}$ with $M_1$ and $M_2$ being two valid redistricting maps. The adjacency matrices $A_1$,$A_2$, and edit distance interpretation of $\dtheta(A_1,A_2)$ are demonstrated. Note that $\dtheta(A_1,A_2)=\theta(1,2)+\theta(3,4)+\theta(1,3)+\theta(2,4)$ which is exactly the minimum sum of edge weights that need to be deleted and added to obtain $A_2$ from $A_1$.}
  \label{fig:ed_fig}
\end{figure*}
A given state is modelled by a graph $G=(V,E)$ where each vertex $v \in V$ represents a voting block (\emph{unit}). Each unit $v$ has a weight $w(v)>0$ which represents its population. Further, $\forall u,v \in V$ there is an edge $e=(u,v) \in E$ if and only if the two vertices are \emph{connected} (geographically this means that units $u$ and $v$ share a boundary). The number of units is $|V|=n$. A \emph{redistricting} (\emph{redistricting map} or simply \emph{map}) $M$ is a partition of $V$ into $k>0$ many districts, i.e., $V=V_1 \cup V_2 \dots \cup V_k$ where each $V_i$ represents a district and $\forall i \in [k], |V_i|\neq 0$ and $\forall i,j \in [k], V_i \cap V_j= \emptyset \ \text{if} \ i \neq j$. The redistricting map $M$ is decided by the induced partition, i.e., $M=\{V_1,\dots,V_k\}$. For a redistricting $M$ to be considered valid, it must satisfy a collection of properties, some of which are specific to the given state. We use the most common properties as stated in  \cite{deford2019recombination,hettle2021balanced}: \textbf{(1)} \emph{Compactness}: The given partitioning should have ``compact'' districts. Although there is no definitive mathematical criterion which decides compactness for districts, some have used common definitions such as Polsby-Popper or Reock Score \cite{alexeev2018impossibility}. Others have used a clustering criterion like the $k$-median objective \cite{cohen2018balanced} or considered the total number of cuts (number of edges between vertices in different districts) \cite{deford2019recombination}. \textbf{(2)} \emph{Equal Population}: To satisfy the desideratum of ``one person one vote'' each district should have approximately the same number of individuals. I.e., a given district $V_i$ should satisfy $\sum_{v \in V_i} w(v) \in [(1-\epsilon)\frac{\sum_{v\in V}w(v)}{k},(1+\epsilon)\frac{\sum_{v\in V}w(v)}{k}]$ where $\epsilon$ is a non-negative parameter relaxing the equal population constraint. \textbf{(3)} \emph{Contiguity}: Each district (partition) $V_i$ should be a connected component, i.e., $\forall i \in [k]$ and $\forall u,v \in V_i$, $v$ should be reachable from $u$ through vertices which only belong to $V_i$.  
%
%
Our proofs do not rely on these properties and therefore can accommodate additional properties if desired. 

Let $\MapsDist$ be the set of all valid maps. Let $\Dist(\MapsDist)$ be a distribution over these maps. Furthermore, define a distance function over the maps $d: \MapsDist \times \MapsDist \rightarrow [0,\infty)$. Then the \emph{population medoid map} is $M^*$ which is a solution to the following: 
\begin{align}\label{pop_med}
    M^* = \argmin\limits_{\map \in \MapsDist} \ \underset{\map' \sim \Dist(\MapsDist)}{\E}[d(\map,\map')]
\end{align}
In words, the population medoid map is a valid map minimizing the expected sum of distances away from all valid maps according to the distribution $\Dist(\MapsDist)$. This serves as a natural way to define a central or most typical map with respect to a given distance metric of interest. 

Since we clearly operate over a sample (a finite collection) from $\Dist(\MapsDist)$; therefore, we assume that the following condition holds: 
\begin{condition}\label{samp_assump}
We can sample maps from the distribution $\Dist(\MapsDist)$ in an independent and identically distributed ($\iid$) manner in polynomial time. 
\end{condition}
We note that although independence certainly does not hold over the sampling methods of \cite{deford2019recombination,mattingly2014redistricting} since they use MCMC methods, it makes the derivations significantly more tractable. Further, the specific choice of the sampling technique is somewhat immaterial to our objective. 

Based on the above condition, we can sample from the distribution $\MapsDist$ efficiently and obtain a finite set of maps $\Maps$ having $T$ many maps, i.e., $|\Maps|=T$.

Now, we define the \emph{sample medoid}, which is simply the extension of the population medoid, but restricted to the given sample. This leads to the following definition:
\begin{align}\label{samp_med}
    \bar{M}^* = \argmin\limits_{\map \in \Maps} \sum_{M' \in \Maps} \ d(\map,\map')
\end{align}

\subsection{Distance over Redistricting Maps}\label{sec:distance_def}
Before we introduce our distance measure, we note that a given map (partition) $M$ can be represented using an ``adjacency'' matrix $A$ in which $A(i,j)=1$ if and only if $\exists V_{\ell} \in M: i,j \in V_{\ell}$ otherwise $A(i,j)=0$. We note that this adjacency matrix can be seen as drawing an edge between every two vertices $i,j$ that are in the same district, i.e., where $A(i,j)=1$. It is clear that we can refer to a map by the partition $M$ or the induced adjacency matrix $A$. Accordingly, we refer to the population medoid as $M^*$ or $\amedpop$ and the sample medoid as $\bar{M}^*$ or $\amed$ 

We now introduce our distance family which is parametrized by a weight matrix $\Theta$ and have the following form:
\begin{align}\label{metric_def}
   \dtheta(A_1,A_2) = \frac{1}{2} \sum_{i,j \in V} \theta(i,j) |A_1(i,j)-A_2(i,j)|
\end{align}
where we only require that $\theta(i,j)>0, \forall i,j \in V$ where $\theta(i,j)$ is the $(i,j)$ entry of $\Theta$. For the simple case where $\theta(i,j)=1, \forall i,j \in V$, our distance $d_{\textbf{1}} (A_1,A_2)$ is equivalent to a Hamming distance over adjacency matrices. When $\theta(i,j)=1, \forall i,j \in V$, we refer to the metric as the \emph{unweighted distance}. We note that such a distance measure was used in previous work that considered adversarial attacks on clustering \cite{chhabra2020suspicion,cina2022black}. 

Another choice of $\Theta$ that leads to a meaningful metric is the \emph{population-weighted distance} where $\theta(i,j)=w(i)w(j)$. This leads to $d_W(A_1,A_2)=\frac{1}{2} \sum_{i,j \in V} w(i)w(j) |A_1(i,j)-A_2(i,j)|$. The population-weighted distance takes into account the number of individuals being separated from one another when vertices $i$ and $j$ are separated from one another\footnote{Recall that each vertex (unit) is a voting block (AKA voter tabulation district) and units may contain different numbers of voters.} by assigning a cost of $w(i)w(j)$. By contrast, the unweighted distance assigns the same cost regardless of the population values and thus has a uniform weight over the separation of units immaterial of the populations which they include.

Another choice of metric which is meaningful, could be of the form $\theta(i,j)=f(l(i,j))$ where $l(i,j)$ is the length of a shortest path between $i$ and $j$ and $f(.)$ is a positive decreasing function such as $f(l(i,j))=e^{-l(i,j)}$. Such a metric would place a smaller penalty for separating vertices that are far away from each other. In general, our distance has an edit distance interpretation. Specifically, if we were to draw edges between vertices according to the entries with $1$ in the adjacency matrix, then given $A_1$ and $A_2$, the distance $\dtheta(A_1,A_2)$ simply equals the minimum total weight (according to $\theta(i,j)$) of the edges that must be added and deleted to obtain $A_2$ from $A_1$. In the case of the unweighted distance, it is precisely equal to the minimum number of edges that have to be deleted from and added to $A_1$ to obtain $A_2$. See Figure \ref{fig:ed_fig} for an illustration.

\section{Justification for Choosing a Central Map}\label{sec:justification}
\paragraph{\textbf{Connection to the Kemeny Rule}:}
We note that the Kemeny rule \cite{kemeny1959mathematics,brandt2016handbook} is the main inspiration behind our proposed framework. More specifically, given a set of alternatives and individuals voting by ranking the alternatives, the Kemeny rule provides a method for aggregating the resulting collection of rankings. This is done by introducing a distance measure over rankings (the Kendall tau distance \cite{kendall1938new}) and choosing the ranking which minimizes the sum of distances away from the other rankings in the collection as the aggregate ranking. 

Although we do not deal with rankings here, we follow a similar approach to the Kemeny rule as we introduce a distance measure over redistricting maps and choose the map which minimizes the sum of the distances as the aggregate map. In fact, recently there has been significant citizen engagement in drawing redistricting maps. For example, in the state of Maryland an executive order from the governor has established a web page to collect citizen submissions of redistricting maps~\cite{MDProposalsProcess}. If each member of a committee was to vote for exactly one map in the given submitted maps, then if we interpret the probability $p_{\map'}$ for a map $\map' \in \MapsDist$ to be the number of votes it received from the total set of votes, then the medoid map $M^*$ (similar to the Kemeny ranking) would be the map which minimizes the weighted sum of distances from the set of maps voted on. We include this result as a proposition and its proof follows directly from the definition we gave above: 
\begin{restatable}{proposition}{medoidkemeny} \label{th:medoid_kemeny}
Suppose we have a committee of $\mathcal{T}$ many voters and that each voter votes for one map from a subset of all possible valid maps $\MapsDist$, then given a map $\map'$, if we assign it a probability $p_{\map'} = \frac{\sum_{\tau =1}^{\mathcal{T}} \nu_{\tau,\map'}}{\mathcal{T}}$ where $\nu_{\tau,\map'} \in \{0,1\}$ is the vote of member $\tau$ for map $\map'$, then the medoid map $M^* = \argmin\limits_{\map \in \MapsDist} \ \underset{\map' \sim p_{\map'}}{\E}[d(\map,\map')]$ is the map that minimizes the sum of distances from the set of valid maps where the distance to each map is weighted by the total votes it receives. 
\end{restatable}

\paragraph{\textbf{Connection to Distance and Clustering Based Outlier Detection}:}
The medoid map, by virtue of minimizing the sum of distances, can be considered a central map. Accordingly, one may consider using the medoid map to test for gerrymandering in a manner similar to distance and clustering based outlier detection \cite{he2003discovering,knox1998algorithms}. More specifically, given a large ensemble of maps, if the enacted map is far from the medoid\footnote{Our experiments use the centroid instead of the medoid map.} in comparison to the ensemble then this suggests possible gerrymandering. In fact, we carry experiments on the states of North Carolina and Pennsylvania (both of which have had enacted maps which were considered gerrymandered) and we indeed find the gerrymandered maps to be faraway whereas the remedial maps are much closer in terms of distance.






%% file: Algorithms.tex
\section{Algorithms}\label{sec:algsandtheory}
We show our linear time algorithm for obtaining the sample medoid in Subsection \ref{sec:samp_medoid}. In Subsection \ref{sec:pop_centoid}, we define the population centroid, derive sample complexity guarantees for obtaining it, and show that its $(i,j)$ entry equals the probability of having $i$ and $j$ in the same district. Finally, in Subsection \ref{sec:pop_medoid}, we discuss obtaining the population medoid and show that in general an arbitrarily large sample is not sufficient to approximate it.

Before we introduce our algorithms, we show that our distance family is a metric (satisfies the properties of a metric):
\begin{restatable}{proposition}{distismetric} \label{th:dist_is_metric}
For all $\Theta$ such that $\forall i,j, \theta(i,j)>0$, the following distance function is a metric.
\[\dtheta(A_1,A_2)=\frac{1}{2}\sum_{i,j \in V} \theta(i,j)|A_1(i,j)-A_2(i,j)|
\]  
\end{restatable}

\subsection{Obtaining the Sample Medoid}\label{sec:samp_medoid}
We note that in general obtaining the sample medoid is not scalable since it usually takes quadratic time \cite{newling2017sub} in the number of samples, i.e. $\Omega(T^2)$. An $O(T^2)$ run time can be easily obtained through a brute-force algorithm which for every map calculates the sum of the distances from other maps and then selects the map with the minimum sum. However, for our family of distances $\dtheta(.,.)$ we show that the medoid map is the closest map to the centroid map (defined below) and show a simple algorithm that runs in $O(T)$ time for obtaining the sample medoid. The fundamental cause behind this speed up is an equivalence between the Hamming distance over binary vectors and the square of the Euclidean distance which is still maintained with our generalized distance. Before introducing the theorem we define $\dthetasq(A_1,A_2)=\frac{1}{2}\sum_{i,j\in V} \theta(i,j)(A_1(i,j)-A_2(i,j))^2$ where the absolute has been replaced by a square. Now we state the decomposition theorem: 
\begin{restatable}{theorem}{decompth} \label{th:decomp_th}
Given a collection of redistricting maps $A_1,\dots,A_T$, the sum of distances of the maps from a fixed redistricting map $A'$ equals the following:
\begin{align}\label{main_decomp_eq}
    \sum_{t=1}^T \dtheta(A_t,A') & =  \sum_{t=1}^T \dthetasq(A_t,\bar{A}_c) + T  \dthetasq(\bar{A}_c,A')
\end{align}
where $\bar{A}_c=\frac{1}{T}\sum_{t=1}A_t$. 
\end{restatable}

Notice that the above theorem introduces the centroid map $\ac$ which is simply equal to the empirical mean of the adjacency maps. It should be clear that with the exception of trivial cases the centroid map $\ac$ is not a valid adjacency matrix, since despite being symmetric it would have fractional entries between $0$ and $1$. Hence, the centroid map also does not lead to a valid partition or districting. Moreover, we note that it is more accurate to call $\ac$ the sample centroid, as opposed to the population centroid $A_c$ (see Subsection \ref{sec:pop_centoid}) which we would obtain with an infinite number of samples. 




\begin{algorithm}[t]
	\begin{algorithmic}
		\STATE{Input: $\Maps=\{A_1,\dots,A_T\}$, $\Theta=\{ \theta(i,j)>0 ,\forall i,j \in V \}$.}
		\STATE{\textbf{1:} Calculate the centroid map $\ac= \frac{1}{T} \sum_{t=1}^T A_t$.}
		\STATE{ \textbf{2:} Pick the map $\bar{A}^* \in \Maps$ which minimizes the $\dthetasq$ distance from the centroid $\ac$, i.e. $\bar{A}^* = \argmin_{A\in \Maps} \dthetasq(A,\ac)$.}
        \RETURN{$\bar{A}^*$} 
	\end{algorithmic}
	\caption{Finding the Sample Medoid}
	\label{alg:fast_pair_alg_weighted}
\end{algorithm}

The above theorem leads to Algorithm \ref{alg:fast_pair_alg_weighted} with the following remark:  
\begin{remark}
Algorithm 1 returns the correct sample medoid and runs in $O(T)$ time. 
\end{remark}
%
We note that calculating the sample medoid in algorithm 1 has no dependence on the generating method. Therefore, if a set of maps are produced through any mechanism and are considered to be representative and sufficiently diverse, then algorithm 1 can be used to obtain the sample medoid in time that is linear in the number of samples. 

\subsection{Sample Complexity for Obtaining the Population Centroid}\label{sec:pop_centoid}
In the prior section, we introduced the sample centroid $\ac$ which is equal to the empirical mean from taking the average of the adjacency matrices, i.e., $\ac =\frac{1}{T} \sum_{t=1}^T A_t$. We now consider the population centroid $\acpop=\lim_{T\rightarrow \infty} \sum_{t=1}^T A_t$. Clearly, by the law of large numbers \cite{zubrzycki1972lectures}, we have $\acpop(i,j)=\E[A(i,j)]$ . It is also clear that $\acpop$ has an interesting property, specifically the $(i,j)$-entry equals the probability that $i$ and $j$ are in the same district:
\begin{restatable}{proposition}{acpopprop} \label{th:acpop_prop}
$\acpop(i,j)=\Pr[i \text{ and } j \text{ in the same district}]$. 
\end{restatable}
Now we show that with a sufficient number of samples, the sample centroid converges to the population centroid entry-wise and in terms of the $\dthetasq$ value. Specifically, we have the following proposition:
\begin{restatable}{proposition}{sampcompone} \label{th:samp_comp_1}
If we sample $T\ge\frac{1}{\epsilon^2}\ln{\frac{n}{\delta}}$ $\iid$ samples, then with probability at least $1-\delta$, we have that $\forall i,j \in V: |\ac(i,j)-\acpop(i,j)| \leq \epsilon$. Further, let $\kappa=\max_{i,j \in V}\sqrt{\theta(i,j)}$, if we have $T \ge \frac{\kappa n^2}{\epsilon} \ln{\frac{n}{\delta}}$ $\iid$ samples, then $\dthetasq(\ac,\acpop)\leq \epsilon$ with probability at least $1-\delta$. 
\end{restatable}

\subsection{Obtaining the Population Medoid}\label{sec:pop_medoid}
Having found the sample centroid $\ac$ and shown that it is a good estimate of the population centroid $\acpop$, we now show that we can obtain a good estimate of the population medoid by solving an optimization problem. Assuming that we have the  population centroid $\acpop$, then the population medoid is simply a valid redistricting map $A$ which has a minimum $\dthetasq(A,\acpop)$ value. This follows immediately from Theorem \ref{th:decomp_th}. More interestingly, we show that this optimization problem is a constrained instance of the min $k$-cut problem:
\begin{restatable}{theorem}{minkcut} \label{th:min_k_cut}
Given the population centroid $\acpop$, the population medoid $\amedpop$ can be obtained by solving a constrained min $k$-cut problem. 
\end{restatable}

If we have a good estimate $\ac$ of the population centroid $\acpop$, then we can solve the above optimization using $\ac$ instead of $\acpop$ and obtain an estimate of the population medoid $\amed$ instead of the true population medoid $\amedpop$ and bound the error of that estimate. The issue is that the min $k$-cut problem is NP-hard \cite{goldschmidt1994polynomial,saran1995finding} \footnote{Note that in our case the min $k$-cut problem can have negative edge weights while the min $k$-cut problem is generally stated with non-negative weights. Nevertheless, we still minimize a cut objective and the non-negative weight min $k$-cut instance is trivially reducible to a min $k$-cut instance with negative and non-negative weights.}. Further, the existing approximation algorithms assume non-negativity of the weights. Even if these approximation algorithms can be tailored to this setting, the additional constraints on the partition being a valid redistricting (each partition being contiguous, of equal population, and compact) make it quite difficult to approximate the objective. In fact, excluding the objective and focusing on the constraint alone, only the work of \cite{hettle2021balanced} has produced approximation algorithm for redistricting maps but has done that for the restricted case of grid graphs. Further, while there exists heuristics for solving min $k$-cut for redistricting maps they only scale to at most around $500$ vertices \cite{validi2022political}.

Having shown the difficulty in obtaining the population medoid by solving an optimization problem, it is reasonable to wonder whether we can gain any guarantees about the population medoid by sampling. We show the negative result that we cannot guarantee that we can estimate the sample medoid of a distribution with high probability by choosing a sampled map even if we sample an arbitrarily large number of maps. This implies as a corollary that the sample medoid does not converge to the population medoid in contrast to the centroid (see Proposition \ref{th:samp_comp_1}).  

\begin{restatable}{theorem}{negativeredistmedoid}\label{th:negative_2d}
For any arbitrary $T$ many $\iid$ samples $\{A_1,\dots,A_T\}$ there exists a distribution over a set of redistricting maps such that: (1) $\Pr[\min_{A\in \{A_1,\dots,A_T\}} d(A,A^*) \ge 0.331] \ge \frac{2}{3}$ and (2) $\Pr[\min_{A \in \{A_1,\dots,A_T\}} f(A) \ge 1.1 f(A^*)] \ge \frac{2}{3}$ where $f(.)$ is the medoid cost function.  
\end{restatable}

We therefore, use a heuristic to find the medoid as mentioned in the next section.

%% file: Experiments.tex
\section{Experiments}\label{sec:experiments}
\begin{figure*}[h!]
  \centering
  \includegraphics[scale=0.7]{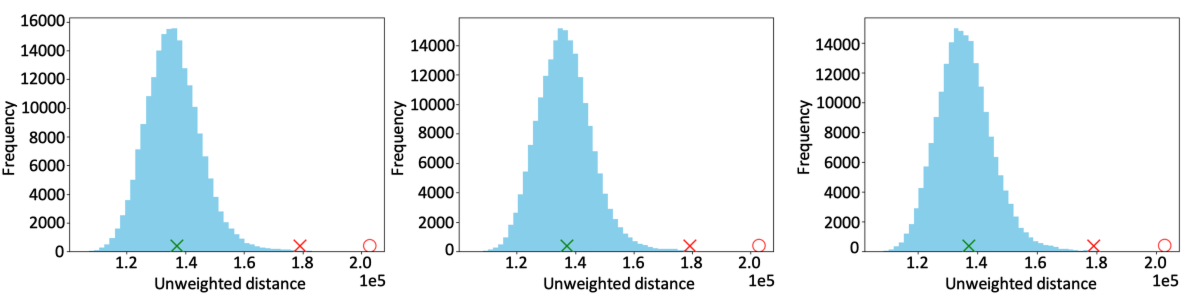}
  \caption{Distance histograms for NC  using the unweighted distance measure. Different plots correspond to different seeds. For NC the distances of gerrymandered maps are indicated with red markers whereas the distances of the remedial maps are indicated with green markers (the circle and the X are for 2011 and 2016 enacted maps, respectively).}
  \label{fig:dist_hist}
\end{figure*}
We conduct experiments on three states, North Carolina (NC), Maryland (MD), and Pennsylvania (PA), which have featured in major court cases on partisan gerrymandering~\cite{LWVPenn2018,CaseRucho2019}. The number of voting units (vertices) are around $2{,}700$, $1{,}800$, and $8{,}900$, and the number of districts are $13$, $8$, and $18$ for NC, MD, and PA, respectively\footnote{For PA, the number of districts was reduced to 17 after the 2020 census, but we are using past election results with 18 districts.}. 
We focus on the results for NC here and discuss the qualitatively-similar results for PA and MD in Appendix \ref{app:exps}. 
NC is especially valuable because we have good examples of both gerrymandered and not gerrymandered maps (enacted maps which are widely considered gerrymandered and a remedial map which is not). 

To generate a collection of maps, we use the Recombination algorithm, $\recom$, from \citet{deford2019recombination} whose implementation is available online. $\recom$ is a Markov Chain Monte Carlo (MCMC) sampling method, and hence, the generated samples are not actually $\iid$. While this means Condition \ref{samp_assump} does not hold, we believe our theorems still have utility and future work can address more realistic sampling conditions. Moreover, we always exclude the first $2{,}000$ samples from any calculation as these are considered to be ``burn-in'' samples\footnote{In MCMC, the chain is supposed to converge to a stationary distribution after some number of steps, called the mixing time. Although the mixing time has not been theoretically calculated for $\recom$, empirically it seems that $2{,}000$ steps are sufficient.}. Throughout this section, when we say distance, we mean $\dthetasq(.,.)$ instead of $\dtheta(.,.)$. Further experimental results and figures are included in Appendix \ref{app:exps}.

\vspace{-0.1cm}
\paragraph{\textbf{Convergence of the Centroid:}} 
Previous work \cite{deford2019recombination,deford2019redistricting} has used the $\recom$ algorithm for estimating statistics such as the histogram of election seats won by a party and determined that using $50{,}000$ samples is sufficient for accurate results. However, our setting is more challenging. Specifically, the centroid includes $\Omega(n)$ entries where $n$ is the total number of voting units (vertices) whereas an election histogram includes only $k$ entries where $k$ is the number of districts---usually orders of magnitude smaller than the number of voting units. Thus, we sample $200{,}000$ maps instead to estimate the centroid. Here, we emphasize the importance of our linear-time algorithm since using a quadratic-time algorithm on samples of the order of even $50{,}000$ could be computationally difficult. Following similar practice to \citet{herschlag2020quantifying} for verifying convergence, we repeat the procedure (sampling using $\recom$ and estimating the centroid) for a total of three times for each state, starting from a different seed map each time and confirming that all three runs result in essentially the same centroid estimate.

To verify the closeness of the different centroid estimates, we calculate the distances between them and compare them to their distances from sampled redistricting maps using $\recom$. We find that the centroids are orders of magnitude closer to each other than to any other sampled map. For example, the maximum unweighted distance between any two centroids is less than $130$ whereas the minimum unweighted distance between any of the three centroids and any sampled map is more than $100{,}000$. 
Similarly, the maximum weighted distance between any two centroids is less than $1.6 \times 10^9$ whereas the minimum weighted distance between a sampled  map and a centroid is at least $1.3\times 10^{12}$ which is again three orders of magnitude higher.

\vspace{-0.1cm}
\paragraph{\textbf{Distance Histogram and Detecting Gerrymandering:}} For each state, we plot the distance histogram from its centroid. More specifically, having estimated the centroid $\ac$, we sample $200{,}000$ maps and calculate $\dthetasq(\ac,A_t)$ where $A_t$ is the $t^{\text{th}}$ sampled map. Figure \ref{fig:dist_hist} shows the unweighted distance histogram for NC\footnote{In Appendix \ref{app:exps}, we show the histogram for PA and MD as well.}. The histogram appears like a normal distribution, peaking at the middle (around the mean) and falling almost symmetrically away. 
This shows that the maps do not concentrate near the centroid even though it minimizes the sum of $\dthetasq(\ac,A_t)$ distances. 
Interestingly, the histogram has a similar shape for both distances  (unweighted and weighted), and this shape remains unchanged across the different seeds.   

Furthermore, previous work used similar sampling methods to detect gerrymandered maps \cite{Chikina2017,mattingly2014redistricting,herschlag2020quantifying}. In essence, these papers demonstrated that the election outcome achieved by the enacted map was rare in comparison to a large sampled ensemble of redistricting maps. Using similar logic, we can also detect gerrymandered maps. Specifically, the 2011 and 2016 enacted maps of NC were widely considered to be gerrymandered, and both maps are at the right tail of the histogram, far from the centroid. By contrast, a remedial NC map drawn by a set of retired judges \cite{herschlag2020quantifying} is much closer to the centroid (see Figure \ref{fig:dist_hist} red and green marked symbols). Interestingly, all gerrymandered maps are in the 99th percentile in terms of distance (for both distance measures and across three seeds). 

This suggests that our method can detect gerrymandered maps with two advantages over previous methods: it does not use election results or partisan outcomes and it is very interpretable. Thus, a guideline or rule that maps should not be far away from the centroid can exist alongside reforms that prohibit explicit partisan consideration. 


\paragraph{\textbf{Finding the medoid:}}
Since we have shown in Subsection \ref{sec:pop_medoid} that the medoid cannot be obtained by sampling, we follow a heuristic that consists of these steps: (1) Sample $200{,}000$ maps and pick the one closest to the centroid $A_{\text{closest}}$. (2) Start the $\recom$ chain from $A_{\text{closest}}$ but given a specific state (redistricting map) we only allow transitions to new states (maps) that are closer to the centroid, and we do this for $200{,}000$ steps to obtain the final estimated medoid $\hat{A}^*$. We follow this procedure three times one for each centroid \footnote{As mentioned before we get three centroids each from sampling a chain that starts with a different seed.}. Figure \ref{fig:NC_medoids} (top row) shows the $A_{\text{closest}}$ medoids from two different runs (each comparing to a different centroid), and it is easy to see they are different. The bottom row shows the final medoids 
$\hat{A}^*$ which are visually more similar to each other and also close in distance.
\begin{figure}
  \centering
  \includegraphics[scale=0.1]{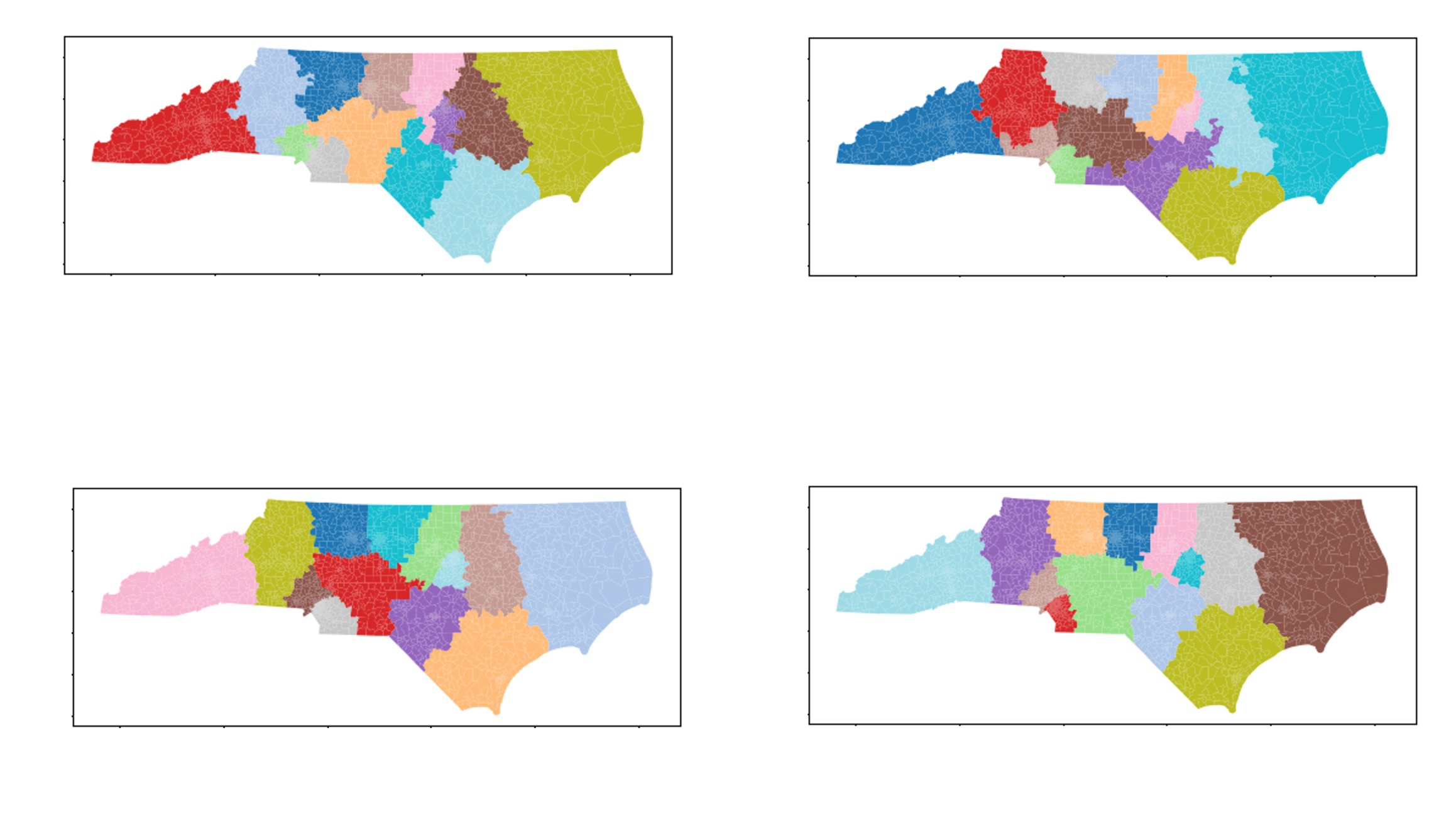}
  \caption{NC medoids, each column is for a specific seed. Top row: $A_{\text{closest}}$, Bottom row: $\hat{A}^*$. }
  \label{fig:NC_medoids}
\end{figure}
\vspace{-0.1cm}

%% file: Conclusion.tex
\section{Conclusion} 
In this paper we introduced a framework which accounts for the challenge of choosing from a large space of valid maps in the redistricting process. Specifically, we introduced a well-motivated family of distance measures and showed how we can obtain the medoid map according to this measure. 
Additionally, we showed experimentally that our framework can be used to find outlier (gerrymandered) maps based on their distance from the centroid map. Finally, we believe that there are further applications of having a distance measure over redistricting maps that are interesting to investigate such as using high dimensional visualization methods to gain further insight into the map ensemble.

%% file: Appendix.tex
\newpage 
\appendix
\onecolumn

\section{Omitted Proof}\label{app:math}
We restate the proposition and give its proof:
\medoidkemeny* 
\begin{proof}
The first step in the proof follows from the  definition of the weighted sum of distances for a map $M$ from the set of all maps: 
\begin{align*}
    M^* & = \argmin\limits_{\map \in \MapsDist} \sum_{M' \in \MapsDist} \Big[ \Big( \sum_{\tau} \nu_{\tau,\map}  \Big) d(M,M') \Big]  \\ 
    & = \argmin\limits_{\map \in \MapsDist} \sum_{M' \in \MapsDist} \Big[ \Big( \frac{\sum_{\tau} \nu_{\tau,\map}}{\mathcal{T}}  \Big) d(M,M') \Big] \\ 
    & = \argmin\limits_{\map \in \MapsDist} \sum_{M' \in \MapsDist} \Big[ p_{\map} d(M,M') \Big] \\ 
    & =  \argmin\limits_{\map \in \MapsDist}  \underset{\map' \sim p_{\map'}}{\E}[d(\map,\map')]
\end{align*}
\end{proof}

We restate the next proposition and give its proof:
\distismetric*
\begin{proof}
Since $\theta(i,j)>0, \forall i,j \in V$, it is clear that $\dtheta(A_1,A_2)$ is non-negative, symmetric, and that it equals zero if and only if  $A_1=A_2$. The triangle inequality follows since $|A_1(i,j)-A_2(i,j)| =|A_1(i,j)-A_3(i,j)+A_3(i,j)-A_2(i,j)| \leq |A_1(i,j)-A_3(i,j)| +|A_3(i,j)-A_2(i,j)|$, therefore we have:
\begin{align*}
    & \dtheta(A_1,A_2)  = \frac{1}{2}\sum_{i,j \in V} \theta(i,j)|A_1(i,j)-A_2(i,j)| \\ 
    & \leq  \frac{1}{2}\sum_{i,j \in V} \theta(i,j) \big(|A_1(i,j)-A_3(i,j)| +|A_3(i,j)-A_2(i,j)|\big) \\ 
    & \leq \frac{1}{2}\sum_{i,j \in V} \theta(i,j) |A_1(i,j)-A_3(i,j)|\\ 
        &\qquad+ \frac{1}{2}\sum_{i,j \in V} \theta(i,j) |A_3(i,j)-A_2(i,j)| \\ 
    & = \dtheta(A_1,A_3) + \dtheta(A_2,A_3)
\end{align*}
\end{proof}

We restate the next theorem and give its proof:
\decompth* 
\begin{proof}
We begin with the following lemma: 
\begin{lemma}\label{ham_eq_l2}
For any $A_1,A_2$ that are binary matrices (entries either $0$ or $1$), with $\dthetasq(A_1,A_2)=\frac{1}{2}\sum_{i,j \in V} \theta(i,j)(A_1(i,j)-A_2(i,j))^2$, then we have that $\dtheta(A_1,A_2)=\dthetasq(A_1,A_2)$. 
\end{lemma}
\begin{proof}
The proof is immediate since $A_1$ and $A_2$ are binary.
\end{proof}
It may seem redundant to introduce a new definition $\dthetasq(.,.)$ since by Lemma \ref{ham_eq_l2}, they are equivalent. However, we will shortly be using $\dthetasq(.,.)$ over matrices which are not necessarily binary, clearly then we might have $\dtheta(A_1,A_2) \neq \dthetasq(A_1,A_2)$. 

We then introduce next lemma:
\begin{lemma}\label{vector_lemma}
For any two matrices (not necessarily binary), the following holds: 
\begin{align}
    \dthetasq(A_1,A_2) = \frac{1}{2} \norm{A^{\Theta}_1-A^{\Theta}_2}^2_2
\end{align}
where $A^{\Theta}_s(i,j)=\sqrt{\theta(i,j)} A_s(i,j), \forall s \in \{1,2\}$ and $\norm{A^{\Theta}_1-A^{\Theta}_2}^2_2$ is the square of the $\ell_2$-norm of the vectorized form of the matrix $(A^{\Theta}_1-A^{\Theta}_2)$.
\end{lemma}
\begin{proof}
\begin{align*}
    \dthetasq(A_1,A_2) & = \frac{1}{2}\sum_{i,j \in V} \theta(i,j)(A_1(i,j)-A_2(i,j))^2 \\ 
    & = \frac{1}{2}\sum_{i,j \in V} (\sqrt{\theta(i,j)}A_1(i,j)-\sqrt{\theta(i,j)}A_2(i,j))^2 \\ 
    & = \frac{1}{2} \norm{A^{\Theta}_1-A^{\Theta}_2}^2_2
\end{align*}
\end{proof}

With the lemmas above, we can now prove the decomposition theorem: 
\begin{align*}
     \sum_{t=1}^T \dtheta(A_t,A')  & = \sum_{t=1}^T \dthetasq(A_t,A')  \ \ \ \ \ \ \text{(using Lemma \ref{ham_eq_l2})} \\ 
    & = \sum_{t=1}^T \frac{1}{2} \norm{A^{\Theta}_t-{A'}^{\Theta}}^2_2  \ \ \ \ \ \ \text{(using Lemma \ref{vector_lemma})} \\ 
    & = \frac{1}{2} \sum_{t=1}^T \norm{A^{\Theta}_t-\acT+\acT-{A'}^{\Theta}}^2_2 \\
    & = \frac{1}{2} \sum_{t=1}^T    \Big[(A^{\Theta}_t-\acT+\acT-{A'}^{\Theta})^\intercal (A^{\Theta}_t-\acT+\acT-{A'}^{\Theta})\Big] \\ 
    & = \sum_{t=1}^T  \frac{1}{2}  \Big[\norm{A^{\Theta}_t-\acT}^2_2 +\norm{\acT-{A'}^{\Theta}}^2_2 \Big]  + \left(\sum_{t=1}^T (A^{\Theta}_t - \acT) \right)^\intercal (\acT-{A'}^{\Theta}) \\
    & =  \sum_{t=1}^T \frac{1}{2} \norm{A^{\Theta}_t-\acT}^2_2 + \frac{T}{2} \norm{\acT-{A'}^{\Theta}}^2_2  + (T\acT-T\acT)^\intercal (\acT-{A'}^{\Theta}) \\ 
    & = \sum_{t=1}^T \dthetasq(A_t,\ac) + T \dthetasq(\ac,A') 
\end{align*}
Note that in the fourth line we take the dot product with the matrices being in vectorized form and that $\acT=\frac{1}{T}\sum_{t=1}^T A^{\Theta}_t$. Note that it follows that $\acT(i,j)=\sqrt{\theta(i,j)}\ac(i,j)$. 
\end{proof}

We restate the next proposition and give its proof:
\acpopprop*
\begin{proof}
By definition of the redistricting adjacency matrices and condition (\ref{samp_assump}) for $\iid$ sampling we have that:
\begin{align*}
    & \acpop(i,j) = \underset{A_t \sim \MapsDist}{\E}[A_t(i,j)]  \\ 
    & = (1) \Pr[i \text{ and } j \text{ are in the same district}]  + (0) \Pr[i \text{ and } j \text{ are not in the same district}] \\
    & = \Pr[i \text{ and } j \text{ are in the same district}]  \\
\end{align*}
\end{proof}

We restate the next proposition and give its proof:
\sampcompone*
\begin{proof}
For a given $i,j \in V$ by the Hoeffding bound we have that: 
\begin{align*}
    &\Pr[|\ac(i,j)-\acpop(i,j)| \leq \epsilon]  \ge 1 - 2 e^{-2\epsilon^2 T}  \\
    & \ge 1-2e^{-2\epsilon^2 \frac{1}{\epsilon^2}\ln{\frac{n}{\delta}}} 
    \ge 1-2(e^{2\ln{\frac{n}{\delta}}} )^{-1} \ge 1-2\frac{\delta^2}{n^2}
\end{align*}
Now we calculate the following event:
\begin{align*}
    & \Pr(\{\forall i,j \in V: |\ac(i,j)-\acpop(i,j)| \leq \epsilon \}) \\
    & = 1 - \Pr(\{\exists i,j \in V: |\ac(i,j)-\acpop(i,j)| > \epsilon \}) \\ 
    & \ge 1- \sum_{i,j \in V} 2 \frac{\delta^2}{n^2} \ge 1 - 2 \delta^2 \frac{(\frac{n^2-n}{2})}{n^2} \ge 1 - \delta^2 \ge 1-\delta \ \ \ \ \ \ \text{( since $\delta \in (0,1)$)}
\end{align*}
Now we prove the second part. By applying the previous result with $\epsilon$ set to $\frac{\sqrt{\epsilon}}{\sqrt{\rho} n}$, we get that with probability at least $1-\delta$,  $|\ac(i,j)-\acpop(i,j)| \leq \frac{\sqrt{\epsilon}}{\sqrt{\rho} n^2}$. It follows that:
\begin{align*}
    \dthetasq(\ac,\acpop) & = \frac{1}{2} \sum_{i,j\in V} \theta(i,j) (\ac(i,j)-\acpop(i,j))^2 \\
    & \leq \frac{1}{2} \sum_{i,j\in V} \theta(i,j) \Big( \frac{\sqrt{\epsilon}}{ \sqrt{\rho} n} \Big)^2 \leq \frac{1}{2} \sum_{i,j\in V} \frac{\epsilon}{n^2} \\
    &\leq \frac{1}{2} \frac{\epsilon}{ n^2} \frac{n^2-n}{2} \leq \epsilon
\end{align*}
\end{proof}

We restate the next theorem and give its proof:
\minkcut* 
\begin{proof}
From Theorem \ref{th:decomp_th}, the population medoid is a valid redistricting map $A$ for which $\dthetasq(A,\acpop)$ is minimized. Note that since $A$ is a redistricting map, unlike $\acpop$ it must be a binary matrix. Therefore, $|A(i,j)-\acpop(i,j)|= (1-\acpop(i,j))+(2\acpop(i,j)-1)(1-A(i,j))$, where this identity can be verified by plugging 0 or 1 for $A(i,j)$ and seeing that it leads to an equality. Define the matrix $B$ as a ``complement'' of $A$. Specifically, $B(i,j)=1-A(i,j)$. It follows that $B(i,j)=1$ if and only if $i$ and $j$ are in different partitions and $B(i,j)=0$ otherwise. Clearly, $B$ is a binary matrix. We can obtain the following:
\begin{align*}
    & \dthetasq(A,\acpop)
     = \frac{1}{2} \sum_{i,j \in V} \theta(i,j) (A(i,j)-\acpop(i,j))^2 \\
    &\quad = \frac{1}{2} \sum_{i,j \in V} \theta(i,j) \big( (1-\acpop(i,j))+(2\acpop(i,j)-1)(1-A(i,j)) \big)^2 \\ 
    &\quad = \frac{1}{2} \sum_{i,j \in V} \theta(i,j) \big( (1-\acpop(i,j))+(2\acpop(i,j)-1)B(i,j) \big)^2 \\ 
    &\quad = \frac{1}{2} \sum_{i,j \in V} \theta(i,j)  \big( (1-\acpop(i,j))^2+2 (1-\acpop(i,j)) (2\acpop(i,j)-1)B(i,j) + (2\acpop(i,j)-1)^2 B^2(i,j) \big)\\ 
    &\quad = \frac{1}{2} \sum_{i,j \in V} \theta(i,j) \big( (1-\acpop(i,j))^2+2 (1-\acpop(i,j)) (2\acpop(i,j)-1)B(i,j) + (2\acpop(i,j)-1)^2 B(i,j) \big)\\ 
    &\quad = \Big[\frac{1}{2} \sum_{i,j \in V} \theta(i,j)  (\acpop^2(i,j)-2 \acpop(i,j) + 1 )  \Big] - \Big[ \frac{1}{2} \sum_{i,j \in V} \theta(i,j) \big(1-2\acpop(i,j) \big) B(i,j)  \Big]
\end{align*} 
 
Note that the first sum in the last equality is a constant and has no dependence on $B$. Hence to minimize $\dthetasq(A,\acpop)$, we maximize the following: 
\begin{align}
    \max_{B} & \sum_{i,j \in V} s(i,j)  B(i,j) \label{eq:minkcut}  \\ 
        \text{s.t. } & B \text{ is a $k$ partition that leads to a valid redistricting map}
\end{align}
where the weight $s(i,j)$ is equal to $s(i,j)=\frac{1}{2} \theta(i,j) \big(1-2\acpop(i,j) \big)$. Clearly, this is a constrained max $k$-cut instance where the partition has to be a valid redistricting map.  

\end{proof}


Theorem \ref{th:negative_2d} shows a negative result for obtaining the population medoid of a set of redistricting maps. Before we show its proof, we show the following theorem which proves a similar result over points in a 2-dimensional Euclidean space. We note that both theorems hold even if the population (true) medoid is sampled with non-zero probability: 
\begin{theorem}\label{th:negappendix}
There exists a distribution over a set of points $\mathcal{P}$ in the 2-dimensional Euclidean space such that given $T$ many $\iid$ samples $\{x_1,\dots,x_T\}$ the following holds: (1) If $T \ge \text{poly}(\epsilon,\ln(1/\rho))$ then with probability at least $1-\rho$ we have $\norm{\frac{1}{T} \sum_{t=1} x_t - \E[x]}^2_2 \leq \epsilon$. However, it is also the case that for any number of samples $T$, we have: (2) $\Pr[\min_{x\in \{x_1,\dots,x_T\}} d(x,x^*) \ge 0.999] \ge \frac{2}{3}$ and (3) $\Pr[\min_{x \in } f(x) \ge 1.2 f(x^*)] \ge \frac{2}{3}$ where $x^*$ is the population medoid and $f(.)$ is the medoid cost function.  
\end{theorem}
\begin{proof}
 \begin{figure}
   \centering
   \includegraphics[scale=0.35]{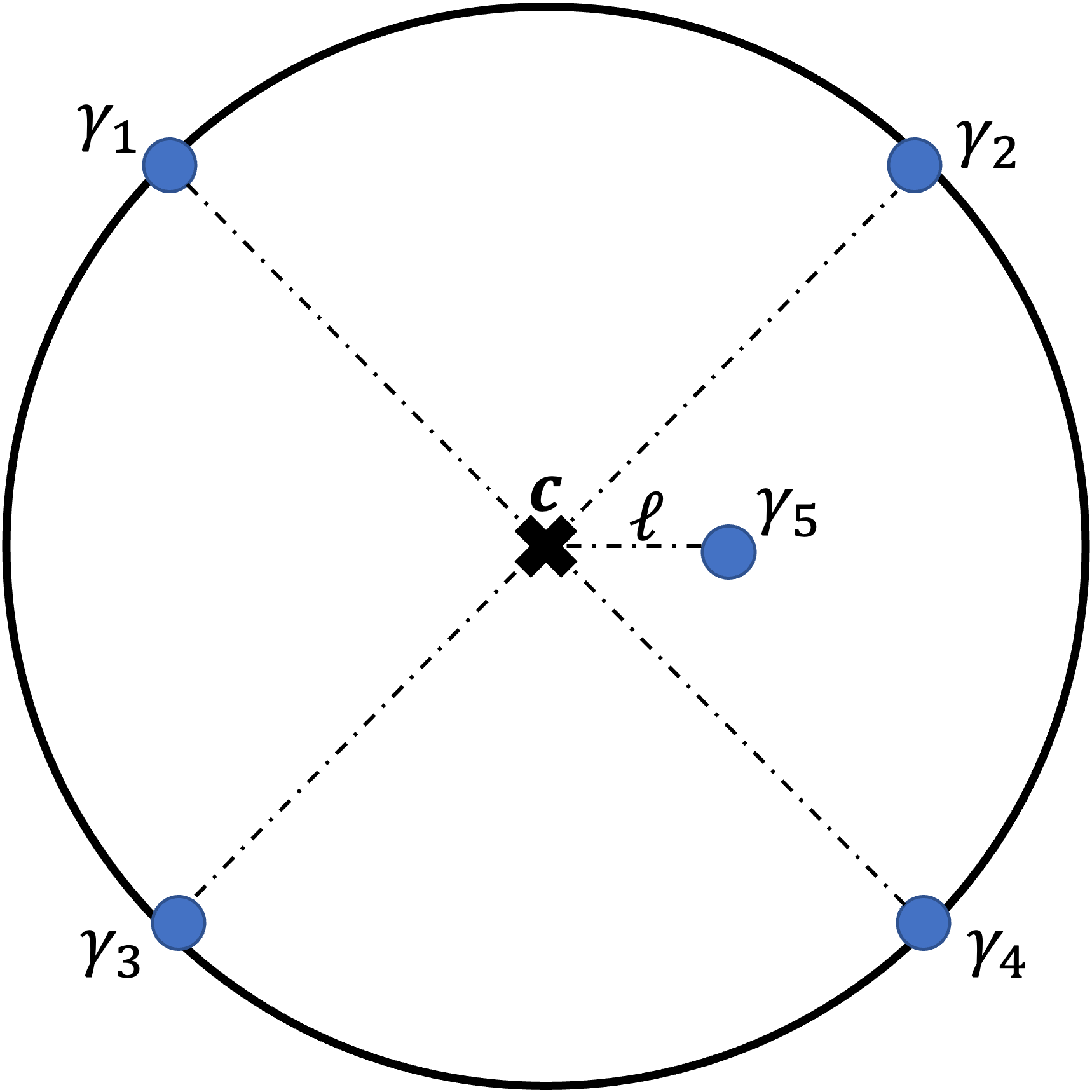}
    \caption{Points $\gamma_1,\gamma_2,\gamma_3$, and $\gamma_4$ all lie on the circle and the angle between any two adjacent points is $90^{\circ}$. Point $\gamma_5$ is a distance $\ell$ from the circle center $c$. The center $c$ is also the origin point of the 2-$d$ plane $(0,0)$.}
    \label{fig:2dvarap}
 \end{figure}

Consider the set of points $\mathcal{P} =\{\gamma_1,\gamma_2,\gamma_3,\gamma_4,\gamma_5\}$ shown in Figure \ref{fig:2dvarap}. Specifically, points $\gamma_1,\gamma_2,\gamma_3$ and $\gamma_4$ lie on a circle of radius $1$ at an equal separation. Point $\gamma_5$ lies closer to the center of the circle $c$ at a distance of $\ell$. Further, each point $\gamma_i$ has a probability $p_i$ of being sampled. Specifically, $p_5=\delta$ and $p_1=p_2=p_3=p_4=\frac{1-\delta}{4}$. 

The centroid of the distribution is the expected value, i.e. $\E[x]= \sum_{j=1}^{5} p_j \gamma_j$. It is straightforward to show that given $\iid$ samples $\{x_1,x_2,\dots,x_T\}$ where $T=\Omega(\frac{\ln(1/\rho)}{\epsilon})$, then $\norm{\frac{1}{T}\sum_{t=1}^T \hat{x}_t - \E[x]}^2_2 \leq \epsilon$ with probability at least $1-\rho$. This follows by applying the generalized Hoeffding inequality (see \cite{ashtiani2016clustering} for example). This proves the first statement.

Now, given a candidate medoid point $\gamma'$, the medoid cost function is $f(\gamma')=\sum_{j=1}^{5} p_j d(\gamma',\gamma_j)$. Note that the medoid is the value of $\gamma'$ which minimizes $f(\gamma')$ and also $\gamma'$ must be an element of the set, i.e. $\gamma' \in \mathcal{P}$. It is easy to verify that $f(\gamma_1)= \frac{1-\delta}{4} [d(\gamma_1,\gamma_1)+d(\gamma_1,\gamma_2)+d(\gamma_1,\gamma_3)+d(\gamma_1,\gamma_4)]+ \delta d(\gamma_1,\gamma_5) > \frac{1-\delta}{4} [0 + 2\sqrt{2} + 2] = \frac{1-\delta}{2} (1+\sqrt{2})$. By symmetry it also follows that $f(\gamma_2),f(\gamma_3),f(\gamma_4)$ are each lower bounded by $\frac{1-\delta}{2} (1+\sqrt{2})$. 

On the other hand, we have $f(\gamma_5) = \frac{1-\delta}{4} \sum_{j=1}^{4} d(\gamma_5,\gamma_j) \leq \frac{1-\delta}{4} \sum_{j=1}^{4} [d(\gamma_5,c) + d(c,\gamma_j)] \leq \frac{1-\delta}{4} (4(1+\ell)) = (1-\delta) (1+\ell)$. Further, for any $i\neq 5: \frac{f(\gamma_i)}{f(\gamma_5)} \ge \frac{\frac{1-\delta}{2} (1+\sqrt{2})}{(1-\delta)(1+\ell)}= \frac{(1+\sqrt{2})}{2(1+\frac{1}{1000})} > 1.2$ by setting $\ell=\frac{1}{1000}$, and therefore $\gamma_5$ is the medoid. Further, $\forall i \in \{1,2,3,4\}: d(\gamma_i,\gamma_5) \ge d(\gamma_i,c)-d(c,\gamma_5) \ge 1-\ell \ge 0.999$. Therefore, to prove the second and third statements it is sufficient to prove that  $\gamma_5$ would be sampled with probability at most $\frac{1}{3}$ in $T$ many $\iid$ samples: 
\begin{align*}
     & \Pr[\text{Sampling $\gamma_5$ in $T$ $\iid$ samples}]  \\
     & = 1- \Pr[\text{Not sampling $\gamma_5$ in $T$ $\iid$ samples}] \\
     & = 1 - (1-\delta)^T \leq \frac{1}{3} \\
     & \iff \delta \leq 1-\sqrt[T]{\frac{2}{3}} 
\end{align*}
Therefore, by setting $\delta \leq 1-\sqrt[T]{\frac{2}{3}}$ the distribution satisfies all three statement simultaneously.   
\end{proof}

Now we restate Theorem \ref{th:negative_2d} and show its proof:
\negativeredistmedoid*
\begin{proof}
\begin{figure}[h]
  \centering
  \includegraphics[scale=0.45]{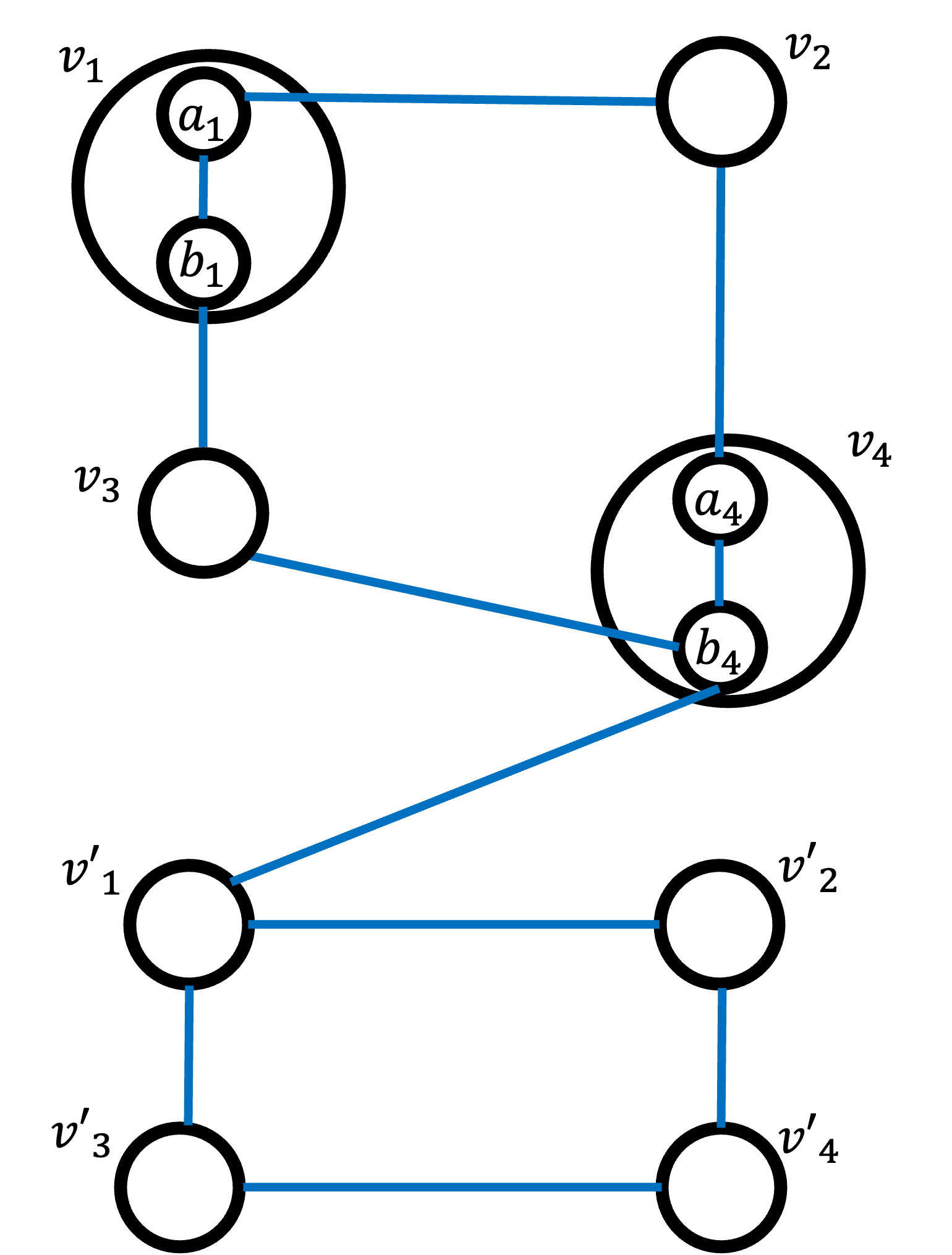}
  \caption{The graph shows a hypothetical state. Blue edges indicate that the vertices are adjacent geographically. All vertices have a weight (population) of 1, except for states $\{a_1,b_1,a_4,b_4\}$ which have a weight of $\frac{1}{2}$.}
  \label{fig:state_hypo_2}
\end{figure}
Consider the hypothetical state shown in Figure \ref{fig:state_hypo_2} where vertices $v_1$ and $v_4$ are further subdividing into two vertices each. We wish to divide the state into $2$ districts ($k=2$). Since each vertex has a weight of $1$, except vertices $\{a_1,b_1,a_4,b_4\}$ which each have a weight of $\frac{1}{2}$, then each district should have a population of $2$ to enforce the equal population rule with a tolerance less than $0.25$. 

Denoting the set of all vertices by $V$ and letting $V'=\{a_1,b_1,a_4,b_4\}$, then the weight parameters of our weighted distance measure are defined as follows: 
\begin{align*}
  \theta(i,j) =
  \begin{cases}
   \epsilon & \text{if $i \ \& \ j \in V'$} \\
   \frac{1}{2} & \text{if $i \in V-V', j \in V'$ or $i \in V', j \in V-V'$} \\
    1 & \text{otherwise}
  \end{cases}
\end{align*}

\begin{figure}[h]
  \centering
  \includegraphics[scale=0.26]{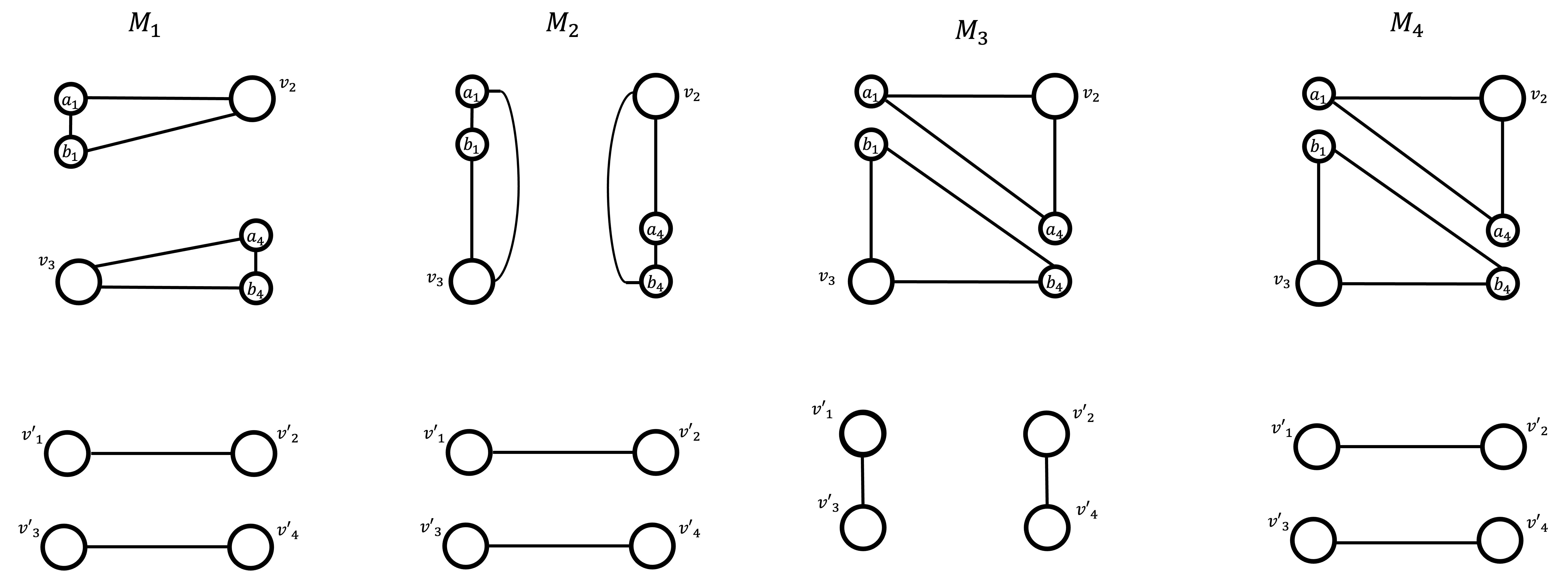}
  \caption{Maps $M_1, M_2, M_3$, and $M_4$. Vertices in the same district are connected with edges.}
  \label{fig:M_all}
\end{figure}
Where $0 < \epsilon\leq 1$. Now, consider the maps $M_1, M_2, M_3$, and $M_4$ shown in Figure \ref{fig:M_all}. Based on the definition of the weighted distance measure, it is not difficult to see that given maps $M_s$ and $M_t$, then $\dtheta(M_s,M_t)$ can be computed visually by drawing the adjacent graphs of $M_s$ and $M_t$ and then finding the minimum number of edges that have to be deleted and added to $M_s$ to produce $M_t$ and adding the weighted $\theta(i.j)$ of these edges. By following this procedure, we can show that $\dtheta(M_1,M_3)=6+4\epsilon$ for example as shown in Figure \ref{fig:M1_M3_dist}. Here we list all distances: 
\begin{align*}
    & \dtheta(M_1,M_2) = 4 \\
    & \dtheta(M_1,M_4) = \dtheta(M_2,M_4) = 2+4\epsilon \\
    & \dtheta(M_1,M_3) = \dtheta(M_2,M_4) = (2+4\epsilon) + 4 = 6+4\epsilon \\
    & \dtheta(M_3,M_4) =  4 \\
\end{align*}
Given a map $M$, the medoid cost function is defined as $f(M)=\sum_{M' \in \MapsDist} p_{M'} d(M,M')$. Let the probabilities for the redistricting maps be assigned as follows: $p_1=p_2=p_3=\frac{1-\delta}{3}$ whereas $\Pr[\text{Sampling a Map $M \notin \{M_1,M_2,M_3\}$}] = \delta>0$. Accordingly, $p_5 \leq \delta$. 

Further, since $\theta(i,j)\leq 1 \ ,\forall i,j \in V$, then maximum distance between any redistricting maps $D$  can be upper bounded by the highest number of edges that can be deleted and added from one map to produce another, therefore $D\leq 2 {|V|\choose 2}=|V|(|V|-1)=10\times9=90$ since $|V|=10$ (see Figure \ref{fig:state_hypo_2}). The medoid cost function can be lower bounded for $M_1,M_2,$ and $M_3$  and upper bounded for $M_4$ as shown below: 
\begin{align*}
    & f(M_1)  > \frac{1-\delta}{3} [d(M_1,M_2) + d(M_1,M_3)] = \frac{1-\delta}{3} (10+4\epsilon) \\ 
    & f(M_2)  > \frac{1-\delta}{3} [d(M_1,M_2) + d(M_2,M_3)] = \frac{1-\delta}{3} (10+4\epsilon) \\ 
    & f(M_3)  > \frac{1-\delta}{3} [d(M_1,M_3) + d(M_2,M_3)] = \frac{1-\delta}{3} (12+8\epsilon) \\ 
    & f(M_4)  < \frac{1-\delta}{3} [d(M_1,M_4) + d(M_2,M_4) + d(M_3,M_4)] + \delta D = \frac{1-\delta}{3} (8+8\epsilon) + 90 \delta \leq  \frac{1-\delta}{3} (9+8\epsilon)\\ 
\end{align*}
Where the last inequality was obtained by setting $\delta < \frac{1}{271}$ since it follows that $90 \delta < \frac{1-\delta}{3}$. From the above bounds it follows that the population medoid cannot cannot be $M_1, M_2$, or $M_3$. 

 Set $\epsilon = \frac{1}{1000}$ and $\delta \leq \frac{1}{1000} <\frac{1}{271}$ and with $1-(p_1+p_2+p_3)=\delta$, then $\frac{1-\delta}{3} (1-4\epsilon)>0.331$ and $\frac{(10+4\epsilon)}{(9+8\epsilon)}>1.11$. With the population medoid denoted by $M^*$, then we have: $\min\limits_{i \in \{1,2,3\}} \frac{f(M_i)}{f(M^*)} \ge \min\limits_{i \in \{1,2,3\}} \frac{f(M_i)}{f(M_4)} \ge \frac{\frac{1-\delta}{3}(10+4\epsilon)}{\frac{1-\delta}{3} (9+8\epsilon)} > 1.11$. Further, it follows by the triangle inequality that $\forall i \in \{1,2,3\}: f(M_i) \leq f(M^*) + d(M_i,M^*)$, thus $d(M_i,M^*) \ge \frac{1-\delta}{3} (1-4\epsilon)$, since otherwise $f(M^*) + d(M_i,M^*) \leq f(M_4) + d(M_i,M^*) < \frac{1-\delta}{3} (9+8\epsilon) + \frac{1-\delta}{3} (1-4\epsilon) = \frac{1-\delta}{3} (10+4\epsilon)$ which would be a contradiction. 
 
 From the above we have shown that, $\forall i \in \{1,2,3\}: d(M_i,M^*) \ge \frac{1-\delta}{3} (1-4\epsilon) > 0.331$ and  $\min\limits_{i \in \{1,2,3\}} \frac{f(M_i)}{f(M^*)} \ge 1.11$, therefore to prove parts (1) and (2) of the theorem it is sufficient to upper bound the probability of sampling a map that is not in  $\{M_1,M_2,M_3\}$ in $T$ \textbf{iid} samples by $\frac{1}{3}$. This leads to the following:
\begin{align*}
    &  \Pr[\text{Obtaining a map $M \notin \{M_1,M_2,M_3\}$ in a given $T$ $\ \iid$ samples}] \\ 
    & = 1 - \Pr[\text{No map $M \in \{M_1,M_2,M_3\}$ in the given $T$ $\ \iid$ samples}] \\ 
    & = 1 - (1-\delta)^T  \leq \frac{1}{3} 
\end{align*}
Therefore, from the above we should have $\delta = \min\{ \frac{1}{1000}, 1-\sqrt[T]{\frac{2}{3}}\}$ to satisfy both parts of the theorem.   


\begin{figure}[h]
  \centering
  \includegraphics[scale=0.3]{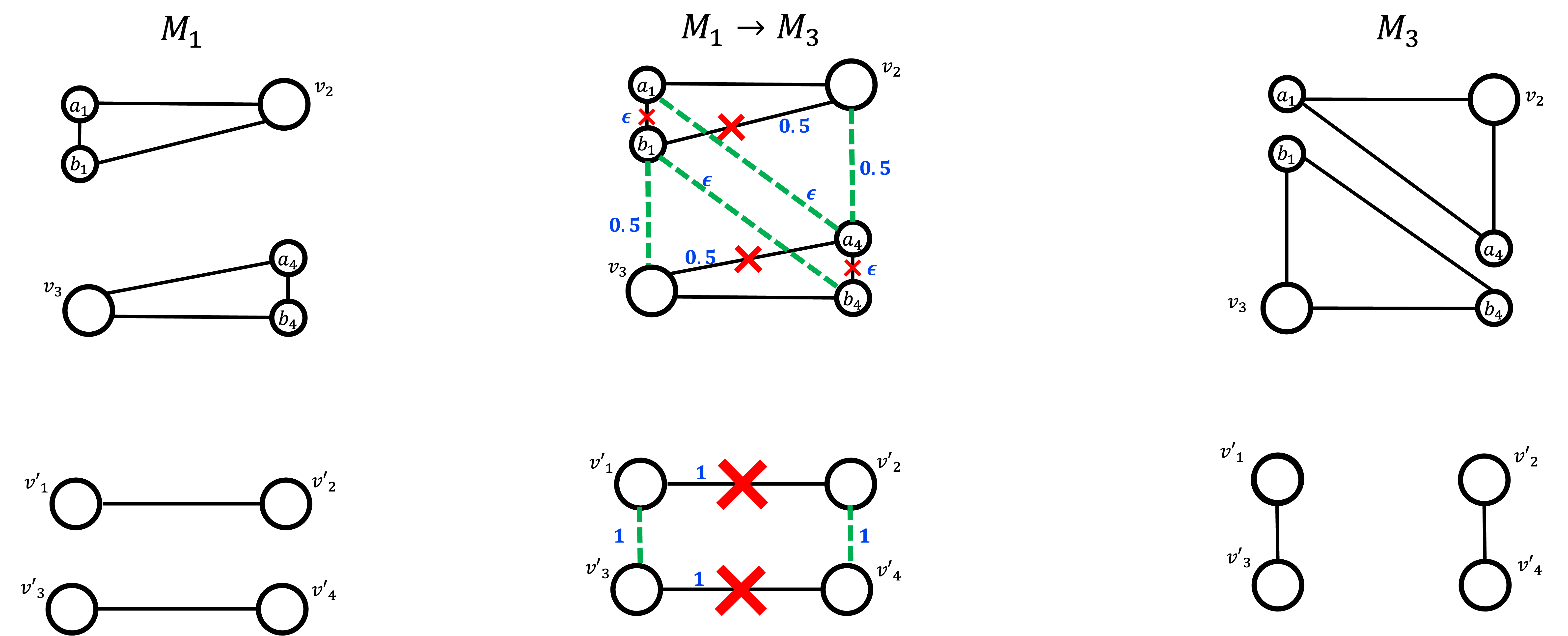}
  \caption{The first map is $M_1$ and the last is $M_3$. The middle map shows the edges the should be deleted from $M_1$ (marked with \textcolor{red}{X}) and the edges should be added to $M_1$ (dashed green edges) to produce $M_3$. The weight of each edge that is deleted or added is shown next to it in blue. By adding the weights we get that $\dtheta(M_1,M_3)=6+4\epsilon$. }
  \label{fig:M1_M3_dist}
\end{figure}

\end{proof}

\paragraph{Remark:} In both theorems \ref{th:negappendix} and \ref{th:negative_2d} the probability of ``failure'' is set to $\frac{2}{3}$ but this is arbitrary as we can make it arbitrarily large by choosing smaller values of $\delta$. But our objective was simply to show that no sampled map would converge to the population medoid or would have a medoid cost function value that converges to the value of medoid cost function of the population medoid. Further, both theorems would hold if the  population medoid is sampled with probability zero, but in our proofs we allowed the population medoid to be sampled with non-zero probability to show that the negative result would still hold even if we were to assume that the population medoid is sampled with non-zero probability.

\input{Exps_Appendix}

\input{details_gerry_detect}

%% file: Exps_Appendix.tex
\section{Additional Experimental Results}\label{app:exps}

\subsection{Convergence to the centroid}\label{subsec:conv_centroid}
 As a reminder to further test the robustness of our results and see the convergence, all calculations are done 3 times for each state, each time starting from a specific seed map\footnote{For MD we use only two seeds and do the calculations twice instead of three times.}. For example, we estimate the centroid three times, by sampling starting from seed maps $s_1$, $s_2$, and $s_3$\footnote{Seed maps are either already enacted maps or produced using optimization functions from the GerryChain toolbox: \textcolor{blue}{https://github.com/mggg/GerryChain}.} and as a result we end up with three centroids $c_1$, $c_2$, and $c_3$. Here we further verify the convergence of the centroids by showing that the final estimated centroids are close to each other. In fact, if we denote the smallest $\dthetasq(.,.)$ distance between a sampled map and any of the centroids $c_1$, $c_2$, or $c_3$ by $\dthetasq_{\min}$. Then we consistently find --for all states (NC,PA,MD) and all distance measures (unweighted and population-weighted)-- that for any two estimated centroids $c_i$ and $c_j$ that $\dthetasq(c_i,c_j) \leq \dthetasq_{\min}$ by \emph{at least $2$ or $3$ orders of magnitude}. This is strong evidence that while the differences in the estimation may be large in value, they are very small relative to the distances between other maps. Therefore, at this scale the estimation error is small. We note moreover the while we used $200,000$ samples to estimate the centroids for NC, we use only $50,000$ for PA and MD in each seed run and interestingly find that $50,000$ samples are sufficient. Table \ref{table:centroid_convergence} shows the maximum $\dthetasq(.,.)$ distance between any two estimated centroid and $\dthetasq_{\min}$. 

\begin{table}[h]
\caption{Maximum Distance Between Any Two Centroids vs Minimum Distance between a Centroid and a Sampled Map} \label{table:centroid_convergence}
\begin{center}
\begin{tabular}{llll}
\textbf{STATE} & \textbf{Distance Measure} & \textbf{$\max\limits_{i,j \in \{1,2,3\}} \dthetasq(c_i,c_j)$} & \ \ \ \ \ \ $\dthetasq_{\min}$ \\
\hline \\
NC         & Unweighted &                   $1.2051\times 10^{2}$   & \ \ \ \ \ \ $1.051030\times 10^{5}$ \\
NC         & Population-Weighted &          $1.5446 \times 10^{9}$ & \ \ \ \ \ \  $1.333014\times 10^{12}$ \\
PA         & Unweighted &                   $5.0794\times 10^{3}$ & \ \ \ \ \ \  $1.099175\times 10^{6}$ \\
PA         & Population-Weighted &          $1.1090\times 10^{10}$ & \ \ \ \ \ \  $ 2.036627\times 10^{12}$ \\
MD         & Unweighted &                   $2.5067\times 10^{2}$ & \ \ \ \ \ \  $5.033550\times 10^{4}$ \\
MD         & Population-Weighted &          $2.7067\times 10^{9}$ & \ \ \ \ \ \  $5.015410\times 10^{11}$ \\
\end{tabular}
\end{center}
\end{table}

\subsection{Distance Histograms and Detecting Gerrymandered Maps}
Here we produce the histogram maps for all states and all distance measures. Each histogram is produced for a given state and distance measure starting from a specific centroid and by sampling from its corresponding seed, e.g. we sample maps starting from seed $s_1$ and calculate $\dthetasq(M,c_1)$ between a sampled map $M$ and centroid $c_1$ to construct one histogram. We find similar observations to Figure \ref{fig:dist_hist}, indeed we see that the shape of the histograms is stable across all states, distances, and centroids. I.e., all of the distance histograms peak in the middle and fall away from the middle similar to a normal distribution. Further, our observations about the gerrymandered maps of NC (Figure \ref{fig:NC_all_histograms}) still hold. In addition, we show similar observations in PA where we find that the gerrymandered map that was struck down by the state supreme court \cite{LWVPenn2018} is an outlier in terms of distance. On the other hand, the remedial map has a much smaller distance and appears near the peak of the histogram. 

Tables \ref{table:nc_unweighted_distances}, \ref{table:nc_weighted_distances}, \ref{table:pa_unweighted_distances}, and \ref{table:pa_weighted_distances} list percentile distances for the above mentioned maps in NC and PA according to both unweighted and weighted distance measures. \emph{We see that the three gerrymander maps (the 2016 and 2011 enacted maps for NC and the 2011 enacted map for PA) are all above the 99th percentile for both distance measures and all three seeds}. On the other hand, the NC Judge's map and PA remedial map are near the peak of the histogram. In addition, the percentiles for these maps are fairly consistent across seeds for any given combination of state and distance measure.

\begin{figure*}[h!]
  \centering
  \includegraphics[scale=0.28]{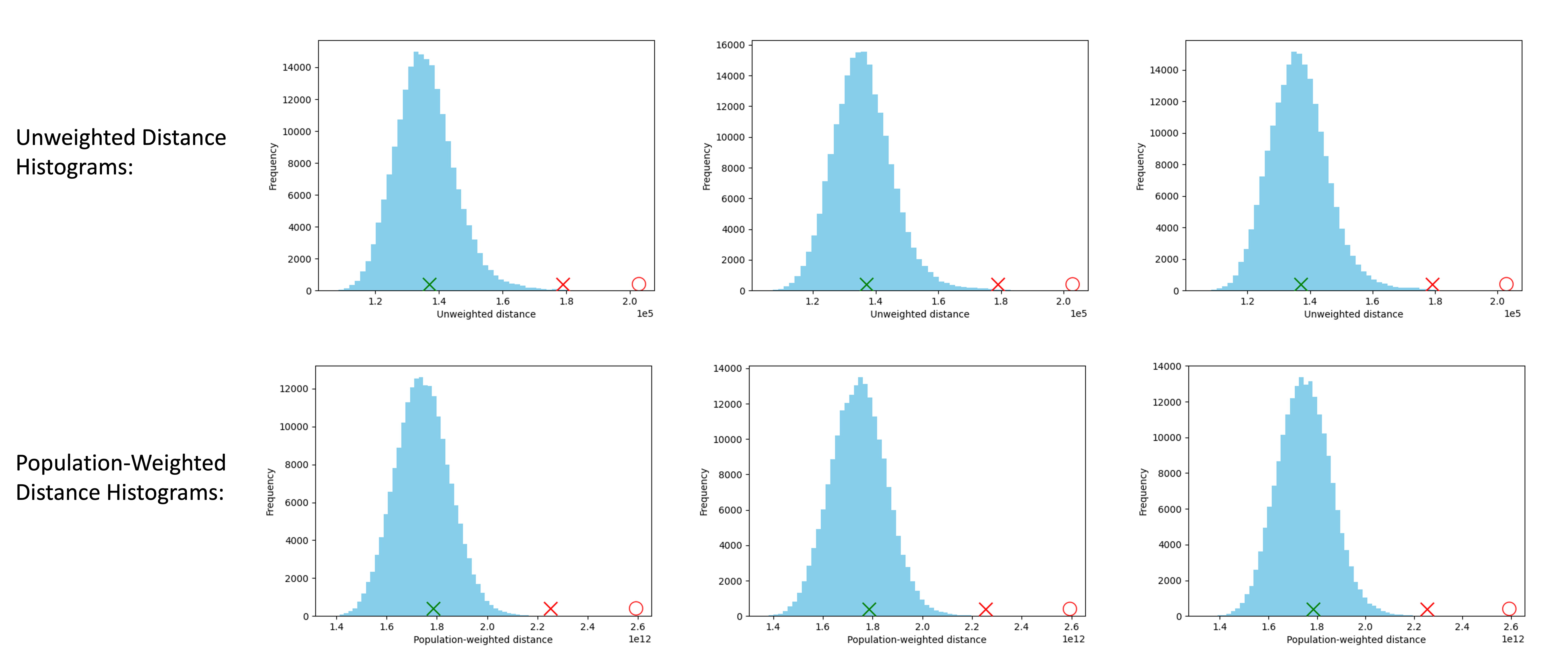}
  \caption{Distance histograms for NC, the distances of gerrymandered maps are indicated with red markers whereas the distances of the remedial maps are indicated with green markers. The $\textcolor{red}{\circ}$ and the \textcolor{red}{X} are for 2011 and 2016 enacted maps, respectively.}
  \label{fig:NC_all_histograms}
\end{figure*}

\begin{figure*}[h!]
  \centering
  \includegraphics[scale=0.28]{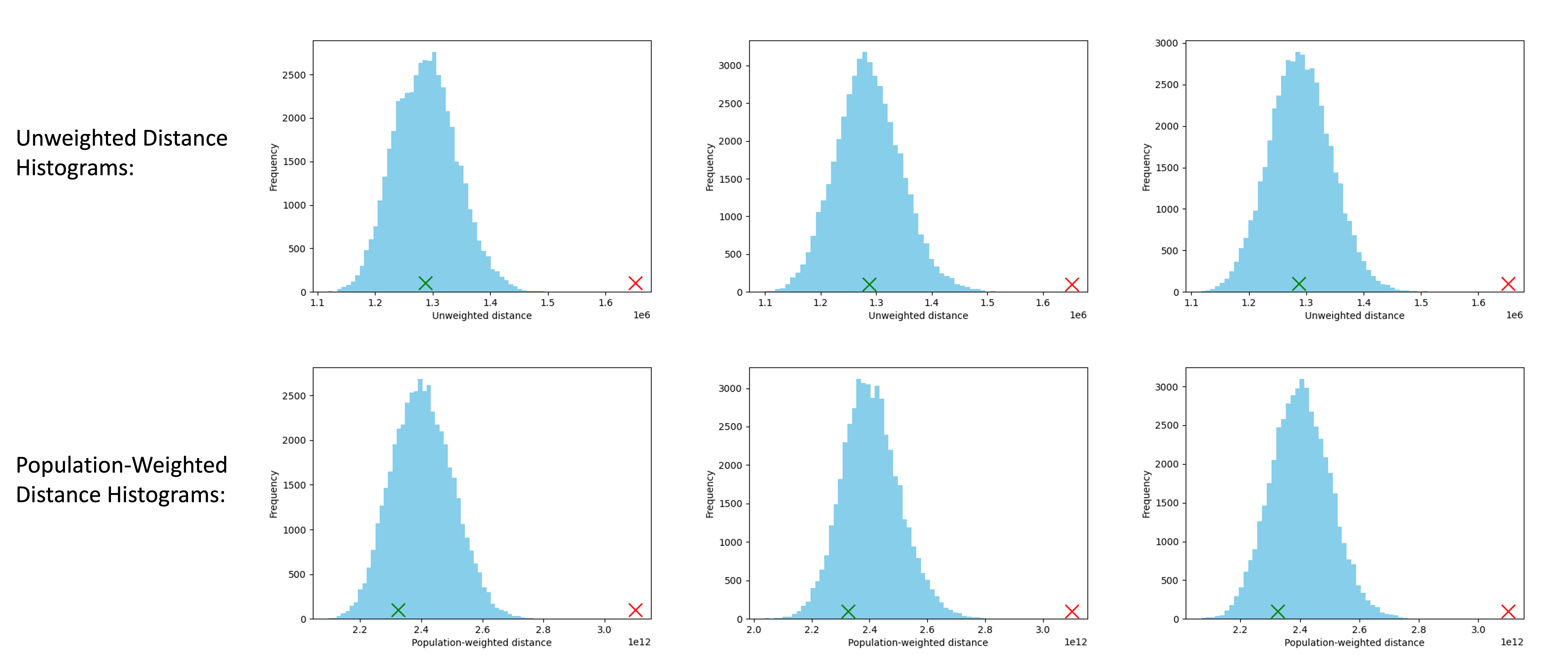}
  \caption{Distance histograms for PA, the distances of gerrymandered maps are indicated with red markers whereas the distances of the remedial maps are indicated with green markers..}
  \label{fig:PA_all_histograms}
\end{figure*}

\begin{figure*}[h!]
  \centering
  \includegraphics[scale=0.3]{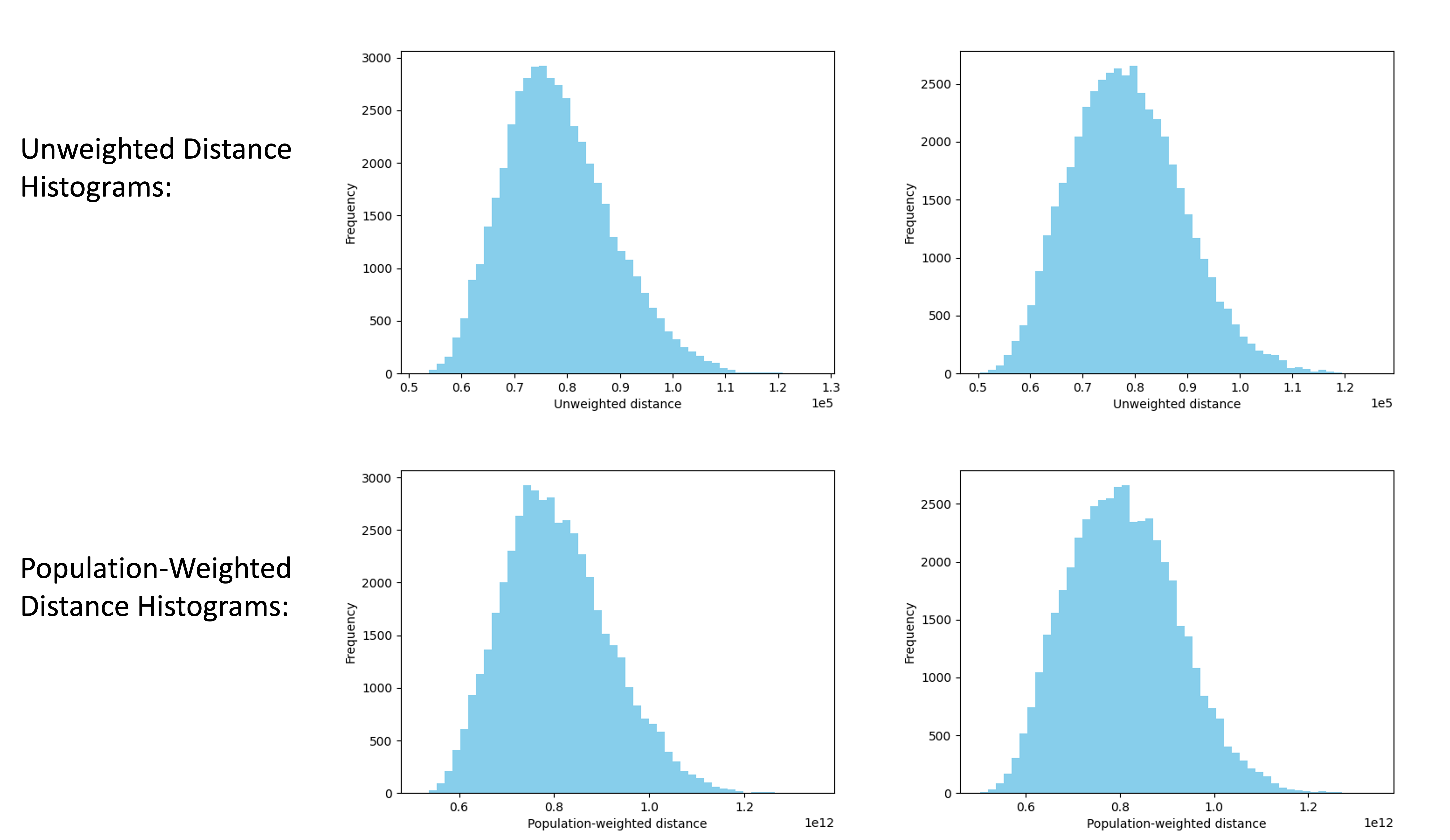}
  \caption{Distance histograms for MD.}
  \label{fig:MD_all_histograms}
\end{figure*}

\begin{table}[h!]
\caption{Percentile distances for special maps in NC using the unweighted distance measure.} \label{table:nc_unweighted_distances}
\begin{center}
\begin{tabular}{*{4}{c}}
\textbf{Starting Seed} & \textbf{Judges' map} & \textbf{2016 enacted map} & \textbf{2011 enacted map} \\
\hline \\
Seed 1         & 58.362 &                  99.942    &	 100.0\\
Seed 2         &  57.004 &         	       99.927     &	 100.0\\
Seed 3         & 54.558 &                  99.943    & 	100.0\\
\end{tabular}
\end{center}
\end{table}

\begin{table}[h!]
\caption{Percentile distances for special maps in NC using the weighted distance measure.} \label{table:nc_weighted_distances}
\begin{center}
\begin{tabular}{*{4}{c}}
\textbf{Starting Seed} & \textbf{Judges' map} & \textbf{2016 enacted map} & \textbf{2011 enacted map} \\
\hline \\
Seed 1         & 65.359 &                 99.999    &	 100.0\\
Seed 2         &  64.454 &         	     99.997     &	100.0\\
Seed 3         & 63.238&                  99.996    & 	100.0\\
\end{tabular}
\end{center}
\end{table}

\begin{table}[h!]
\caption{Percentile distances for special maps in PA using the unweighted distance measure.} \label{table:pa_unweighted_distances}
\begin{center}
\begin{tabular}{*{3}{c}}
\textbf{Starting Seed} & \textbf{Remedial map} & \textbf{2011 enacted map} \\
\hline \\
Seed 1         & 49.992 &           100.0    \\
Seed 2         &  51.875 &         	100.0  \\
Seed 3         &  48.715 &                  100.0    \\
\end{tabular}
\end{center}
\end{table}

\begin{table}[h!]
\caption{Percentile distances for special maps in PA using the weighted distance measure.} \label{table:pa_weighted_distances}
\begin{center}
\begin{tabular}{*{3}{c}}
\textbf{Starting Seed} & \textbf{Remedial map} & \textbf{2011 enacted map} \\
\hline \\
Seed 1         & 22.383 &            100.0    \\
Seed 2         &  20.956 &         	100.0    \\
Seed 3         & 21.813 &                 100.0    \\
\end{tabular}
\end{center}
\end{table}

\subsection{Finding the Medoid}
Similar to the previous experiments, we produce different medoids, one for each state, seed, and distance measure. Having chosen the state and distance measure (unweighted or population-weighted), then for a given centroid $c_i$ and its corresponding seed $s_i$, we start by sampling $T_1$ maps from the seed and pick the sampled map $M$ whose distance $\dthetasq(M,c_i)$ is the smallest. We refer to this map as the \emph{initial medoid}. Clearly, we would have 3 such maps for each state and distance measure. Having found the initial map, we start sampling again but this time starting from the initial map. However, we modify the transition in the $\recom$ chain. Specifically, if we are at a state (map) $M$, then we only transition to a new state (map) $M'$ if its closer to the centroid, i.e. $\dthetasq(M',c_i) < \dthetasq(M,c_i)$. Doing this for $T_2$ iterations, we obtain the \emph{final medoid}. Again we would have one final medoid map for each state, seed, and distance measure. We use $T_1=T_2=200,000$ for NC whereas for PA and MD we have $T_1=50,000$ and $T_2=15,000$. Empirically, we find that $T_1=50,000$ and $T_2=15,000$ are sufficient (see Table \ref{table:re_measure} and the rest of this section for more discussion).

\paragraph{Relative Error Measure:} We have seen in Subsection \ref{subsec:conv_centroid} that we have convergence in the centroid and therefore we assume that we are dealing with one centroid $\ac$ and that $\ac$ is a very good approximation of the population centroid $\acpop$. Since the medoid is supposed to be the closest valid map to the centroid, for any two given mediods $\mh{1}$ and $\mh{2}$ we calculate the relative error between them as follows: 
\begin{align*}
    \RE(M_1,M_2) = \frac{|\dthetasq(\mh{1},\ac)-\dthetasq(\mh{2},\ac)|}{\min\{\dthetasq(\mh{1},\ac),\dthetasq(\mh{2},\ac)\}}
\end{align*}
We use the relative error measure $\RE(.,.)$ to see how the final medoids approximate the medoid cost function relative to one another.

\paragraph{General Observations and Conclusion:} Here we note the observations we reach from the following subsections: (1) For a given state and seed, the initial medoids are the same regardless of distance used (unweighted or population-weighted). (2) For NC and MD, the final medoids are visually very similar, however for PA which is much larger we notice significant differences. (3) Medoid maps can lead to different election results even if they are visually very similar and very close in terms of distance $\dtheta$. 
(4) Medoid maps do not lead to rare election outcomes and the election outcome is at the peak of the election histogram or close to it. (5) The relative error measure $\RE(.,.)$ over final medoid is in general very small, at most equal to $3.87\%$. The fact that we find final medoids with small relative error but that are still different may be explained by the fact that the medoid (or approximate medoids) may not be unique. Regardless, we believe the final medoids maps we obtain are useful as starting maps that can be refined to draw final redistricting maps to be enacted.




\subsubsection{Initial Medoids and Final Medoids}
Figures \ref{fig:NC_all_medoids}, \ref{fig:PA_all_medoids}, and \ref{fig:MD_all_medoids} show the initial and final medoids for NC, PA, and MD, respectively. We note that in all figures, that maps in the first column are all using seed $s_1$ and centroid $c_1$, the second column are using seed $s_2$ and centroid $c_2$, and so on. Further, the initial medoid maps are found to be the same for both distances (unweighted or population-weighted). In general, we notice that the NC and MD final medoid maps are much similar to one another than PA as there are large regions in PA that are clearly assigned to different districts.    





\begin{figure}
  \centering
  \includegraphics[scale=0.38]{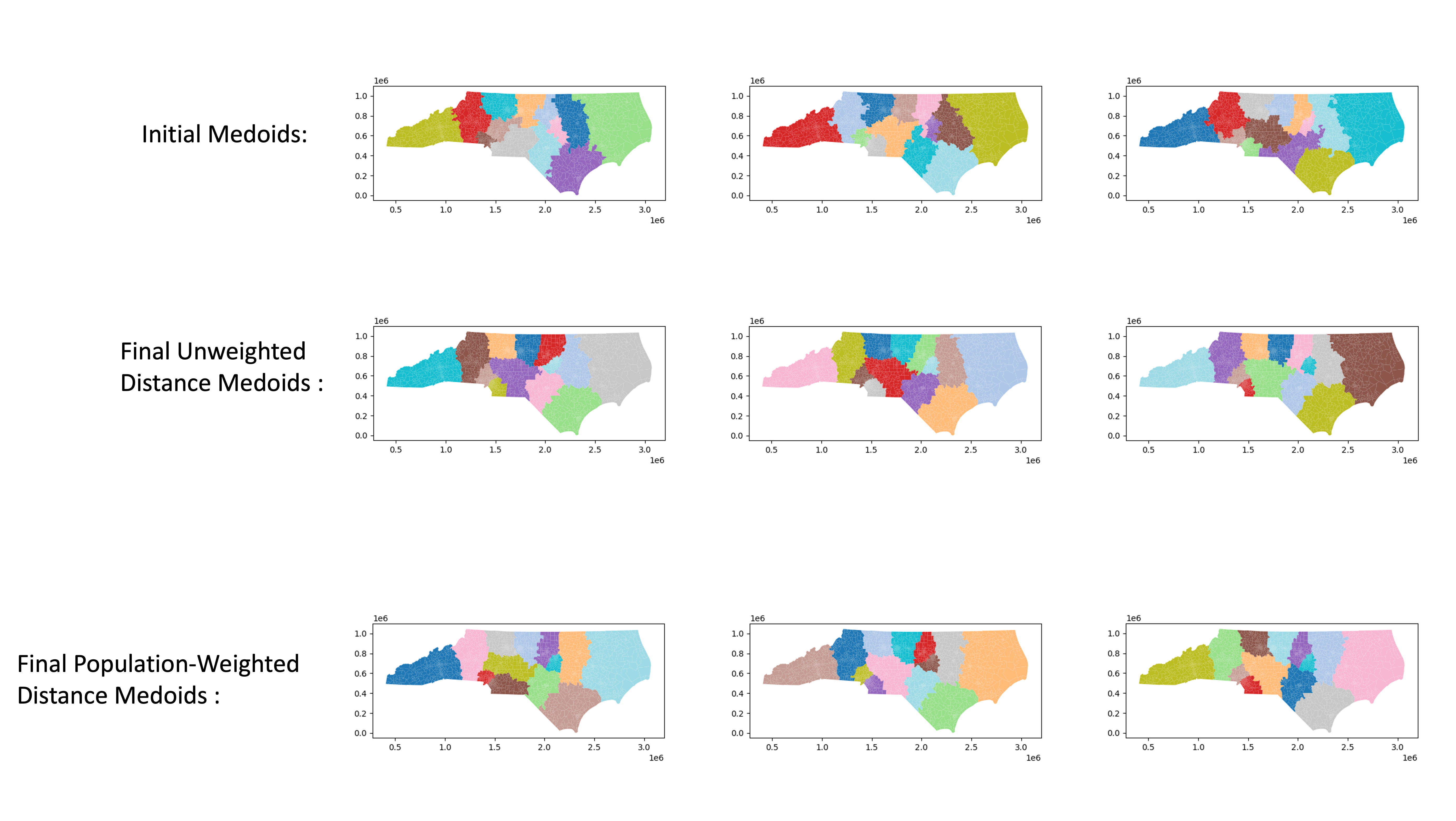}
  \caption{NC Medoids. (Top Row): Initial Medoids, (Middle Row): Final Unweighted Distance Medoids, and (Bottom Row): Final Population-Weighted Distance Medoids. Each column is for a specific seed and its associated centroid.}
  \label{fig:NC_all_medoids}
\end{figure}

\begin{figure}
  \centering
  \includegraphics[scale=0.38]{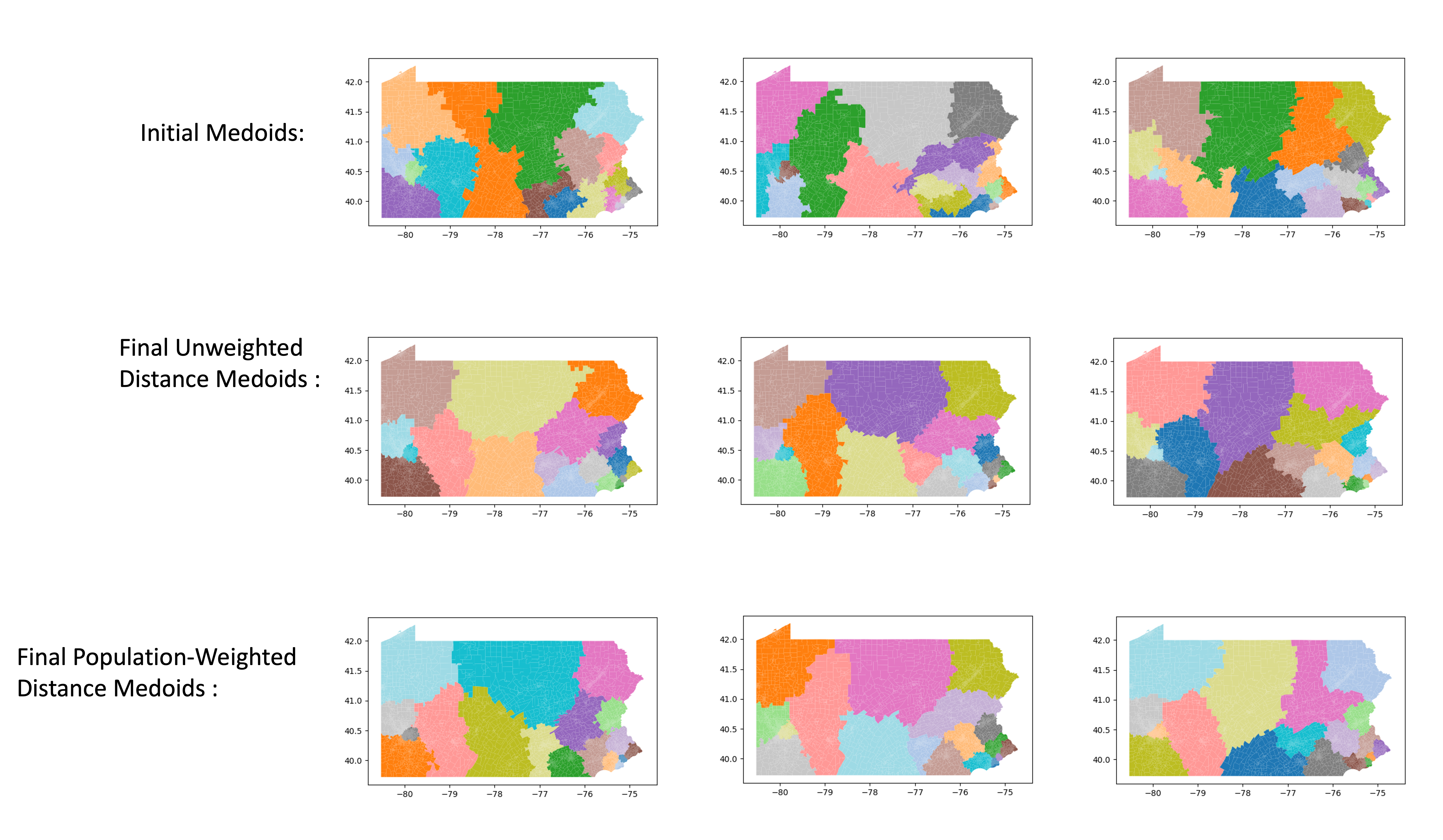}
  \caption{PA Medoids. (Top Row): Initial Medoids, (Middle Row): Final Unweighted Distance Medoids, and (Bottom Row): Final Population-Weighted Distance Medoids. Each column is for a specific seed and its associated centroid.}  \label{fig:PA_all_medoids}
\end{figure}

\begin{figure}
  \centering
  \includegraphics[scale=0.38]{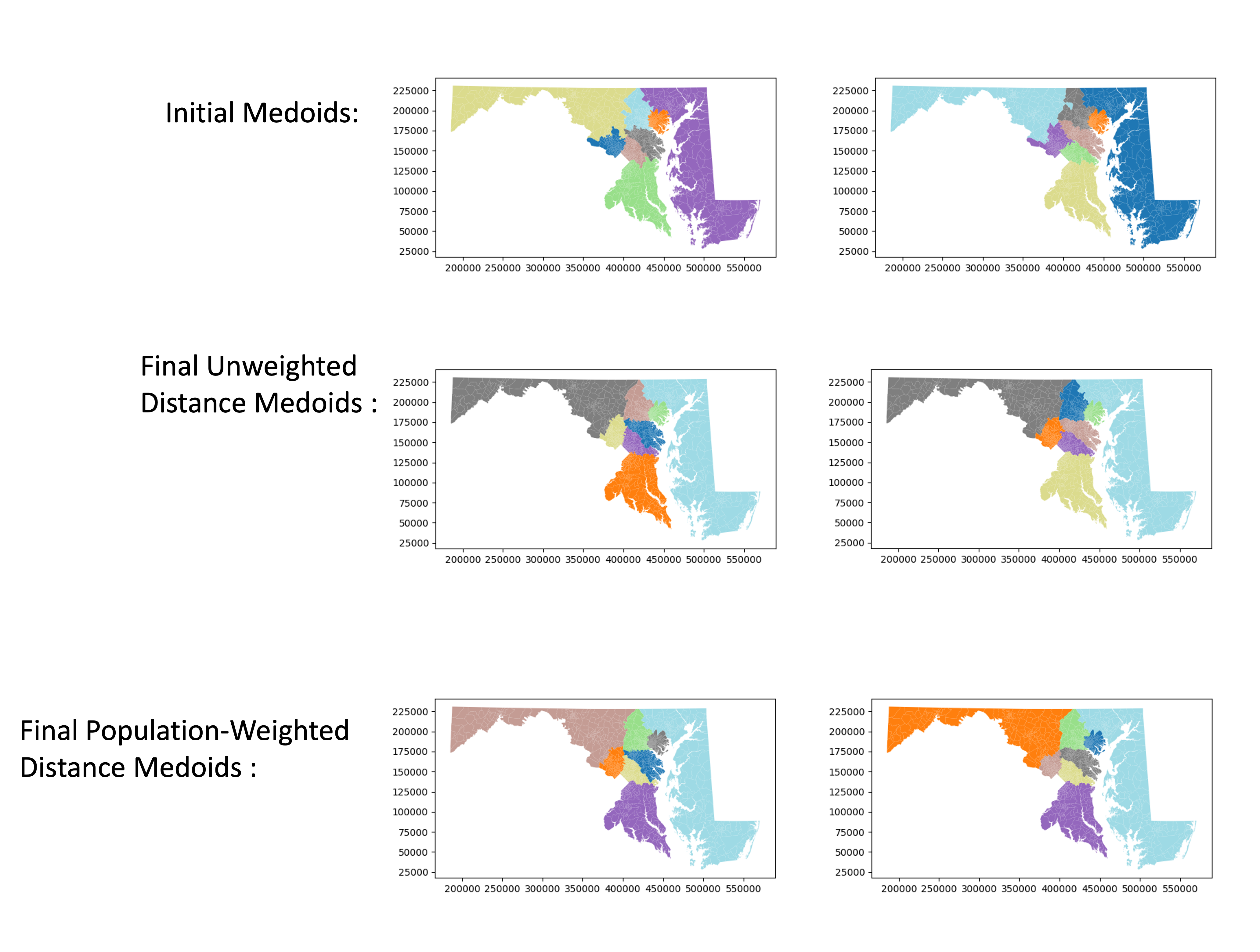}
  \caption{MD Medoids. (Top Row): Initial Medoids, (Middle Row): Final Unweighted Distance Medoids, and (Bottom Row): Final Population-Weighted Distance Medoids. Each column is for a specific seed and its associated centroid.}
  \label{fig:MD_all_medoids}
\end{figure}

\subsubsection{Relative Error $\RE(.,.)$ between the Final Medoids}
Since we have three final medoids for each state and distance measure, we report the maximum relative error value that can be obtained using any pair of maps. Table \ref{table:re_measure} shows the maximum relative error percentage value, clearly the largest relative error value is at most around $3.87\%$. 
\begin{table}[h!]
\caption{Maximum Relative Error $\RE$ between final medoids for each state and distance measure} \label{table:re_measure}
\begin{center}
\begin{tabular}{lll}
\textbf{STATE} & \textbf{Distance Measure} & \textbf{Max $\RE(.,)$ Percentage Value Over Final Medoid Pair} \\
\hline \\
NC         & Unweighted &                   \ \ \ \ \ \ \ \ \ \ \ \ $0.1701$   \\
NC         & Population-Weighted &          \ \ \ \ \ \ \ \ \ \ \ \ $3.2506$   \\
PA         & Unweighted &                   \ \ \ \ \ \ \ \ \ \ \ \ $1.3482$   \\
PA         & Population-Weighted &          \ \ \ \ \ \ \ \ \ \ \ \ $3.5210$   \\
MD         & Unweighted &                   \ \ \ \ \ \ \ \ \ \ \ \ $3.6167$   \\
MD         & Population-Weighted &          \ \ \ \ \ \ \ \ \ \ \ \ $3.8749$   \\
\end{tabular}
\end{center}
\end{table}


\subsubsection{Election Histograms and Medoid Election Results}
Here we show the election histogram for each state using votes for a specific election. First, note according to the distribution of votes in each district the number of seats won by a specific party can be calculated. Since we have two-party results, we show only the seats for one party (Democratic party). Finding election histograms using sampling methods is well-established \cite{deford2019recombination,herschlag2020quantifying} therefore we show only one histogram for each state, see Figures \ref{fig:NC_election_hist},\ref{fig:PA_election_hist}, and \ref{fig:MD_election_hist}. We note that the election histograms are the same independent of the choice of seed (as expected). Further, election histograms are not related to a centroid or a chosen distance measure. For each state and distance measure, we record the number of seats for the final medoid maps as shown in the election tables below. In general, medoid maps lead to election results that have high probability, the outcome with the highest probability or close to it. Moreover, we find that the election outcomes of different medoids can significantly overlap. 


\begin{figure}[h!]
  \centering
  \includegraphics[scale=0.4]{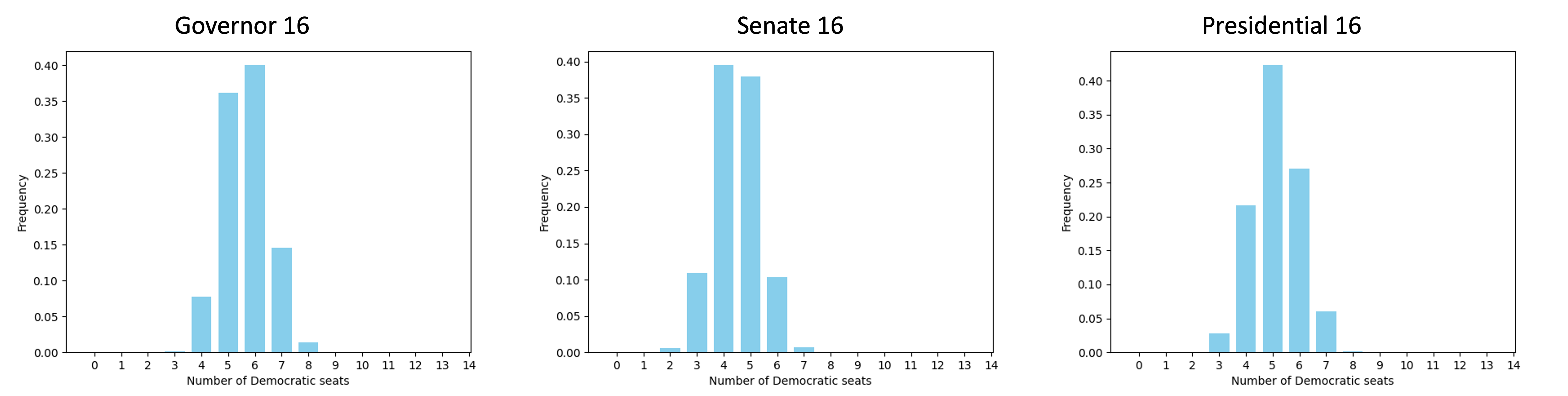}
  \caption{NC Election Histograms.}
  \label{fig:NC_election_hist}
\end{figure}

\begin{figure}[h!]
  \centering
  \includegraphics[scale=0.33]{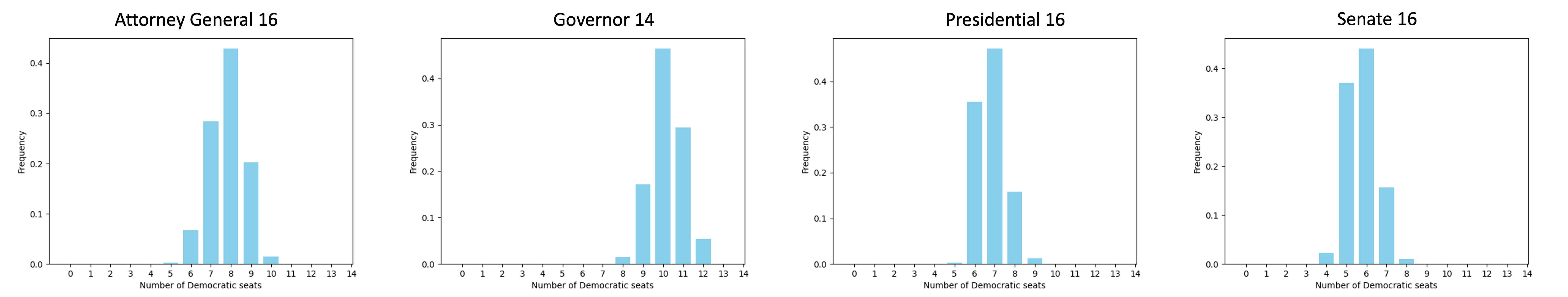}
  \caption{PA Election Histograms.}
  \label{fig:PA_election_hist}
\end{figure}

\begin{figure}[h!]
  \centering
  \includegraphics[scale=0.33]{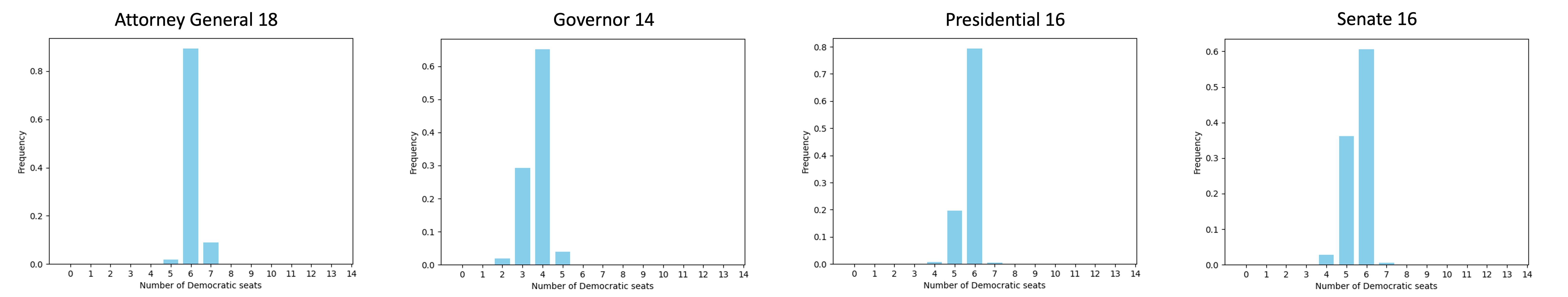}
  \caption{MD Election Histograms.}
  \label{fig:MD_election_hist}
\end{figure}

\begin{table}[h!]
\caption{Election results for the unweighted NC medoids.}
\begin{center}
\begin{tabular}{*{4}{c}}
\textbf{Starting Seed} & \textbf{Governor 16} & \textbf{Senate 16} & \textbf{Presidential 16} \\
\hline \\
Seed 1         & 5 &                 4    &	 5\\
Seed 2         &  6 &         	     4     &	5\\
Seed 3         & 6&                  4    & 	5\\
\end{tabular}
\end{center}
 \label{NC_un_elec)}
\end{table}

\begin{table}[h!]
\caption{Election results for the weighted NC medoids.} \label{table:NC_weighted_medoids_election)}
\begin{center}
\begin{tabular}{*{4}{c}}
\textbf{Starting Seed} & \textbf{Governor 16} & \textbf{Senate 16} & \textbf{Presidential 16} \\
\hline \\
Seed 1         & 5 &                 4    &	 4\\
Seed 2         &  6 &         	     4     &	5\\
Seed 3         & 6&                  4    & 	6\\
\end{tabular}
\end{center}
\end{table}

\begin{table}[h!]
\caption{Election results for the unweighted PA medoids.} \label{table:PA_unweighted_medoids_election)}
\begin{center}
\begin{tabular}{*{5}{c}}
\textbf{Starting Seed} & \textbf{Attorney General 16} & \textbf{Governor 14} & \textbf{Presidential 16} & \textbf{Senate 16} \\
\hline \\
Seed 1         & 8 &                 10    &	 6   &	 6\\
Seed 2         &  8 &         	     9     &	7 &	 6\\
Seed 3         & 7&                  10    & 	6 &	 6\\
\end{tabular}
\end{center}
\end{table}

\begin{table}[h!]
\caption{Election results for the weighted PA medoids.} \label{table:PA_weighted_medoids_election)}
\begin{center}
\begin{tabular}{*{5}{c}}
\textbf{Starting Seed} & \textbf{Attorney General 16} & \textbf{Governor 14} & \textbf{Presidential 16} & \textbf{Senate 16} \\
\hline \\
Seed 1         & 8 &                 9    &	 7   &	 5\\
Seed 2         &  8 &         	     9     &	7 &	 6\\
Seed 3         & 8&                  10    & 	6 &	 6\\
\end{tabular}
\end{center}
\end{table}

\begin{table}[h!]
\caption{Election results for the unweighted MD medoids.} \label{table:MD_unweighted_medoids_election)}
\begin{center}
\begin{tabular}{*{5}{c}}
\textbf{Starting Seed} & \textbf{Attorney General 18} & \textbf{Governor 14} & \textbf{Presidential 16} & \textbf{Senate 16} \\
\hline \\
Seed 1         & 6 &                 4    &	 6   &	 5\\
Seed 2         &  6 &         	     4     &	6 &	 6\\
\end{tabular}
\end{center}
\end{table}

\begin{table}[h!]
\caption{Election results for the weighted MD medoids.} \label{table:MD_weighted_medoids_election)}
\begin{center}
\begin{tabular}{*{5}{c}}
\textbf{Starting Seed} & \textbf{Attorney General 18} & \textbf{Governor 14} & \textbf{Presidential 16} & \textbf{Senate 16} \\
\hline \\
Seed 1         & 6 &                 4    &	 6   &	 5\\
Seed 2         &  6 &         	     4     &	6 &	 6\\
\end{tabular}
\end{center}
\end{table}

%% file: details_gerry_detect.tex
\clearpage
\section{Further Details on Gerrymandering Detection} \label{app:gerry_detect}
\subsection{Gerrymandering Interpretation of Distance Outlier Maps}
What does it mean for a map $\ag$ to a have very large distance $\dthetasq(\ag,\ac)$ from the centroid? Here, we show two interesting and mathematically justified interpretations which support the claim that such a map $A_G$ should be considered gerrymandered. For ease of representation, we assume that we are in the asymptotic regime, i.e. $\ac=\acpop$:  
\paragraph{(1) $\ag$ is Highly Dissimilar from the Ensemble.} 
We show that the distance from the centroid is a direct measure of the average distance from the ensemble of maps. This follows by direct manipulation of the decomposition theorem \ref{th:decomp_th}: 
\begin{align}
\dthetasq(\ag,\ac) = \underbrace{\frac{1}{T} \sum_{t=1}^T \dtheta(\ag,A_t)}_{\text{Average distance between map $\ag$ and the ensemble}} -  \underbrace{\frac{1}{T}\sum_{t=1}^T \dthetasq(A_t,\ac)}_{\text{Constant independent of the map $\ag$}}
\end{align} 
Therefore, the above equation shows that a map $\ag$ with a large distance from the centroid has to be highly dissimilar from the ensemble. This gives further justification for why maps that are outliers from the centroid such as those  of NC and PA (marked in the histograms of Figures \ref{fig:NC_all_histograms} and \ref{fig:PA_all_histograms}) should be considered to be gerrymandered. 
It also helps explain why maps which are outliers in terms of partisan election results would also tend to be outliers according to our non-partisan metric. Intuitively, maps which produce rare election results with respect to the ensemble should have rare structure with respect to the ensemble. Ideally, the metric we propose is then able to detect abnormal maps which might disadvantage groups beyond those harmed by the typical partisan or racial gerrymandering we know to look for. 

Importantly, we note a \emph{significant computational advantage}, since in general finding the average distance of a map $\ag$ from $T$ sampled maps would require $\Omega(T)$ time, but with the above equation we only need to find the distance between the map $\ag$ and the centroid.

\paragraph{(2) $\ag$ Separates Same District Vertices and Unites Separated Vertices.} Proposition \ref{th:acpop_prop} shows that $\acpop(i,j)=\Pr[i \text{ and } j \text{ in the same district}]$, i.e., the probability that vertices $i$ and $j$ are in the same district equals $\acpop(i,j)$. Therefore consider the unweighted $\dthetasq$ distance, i.e. $d_2(\ag,\ac) = \frac{1}{2} \sum_{i,j\in V} (\ag(i,j)-\ac(i,j))^2$. It follows that if $d_2(\ag,\ac)$ is abnormally large, then $\ag$ is following a different distribution where on average if vertices $i$ and $j$ were usually in the same district (i.e., $\Pr[i \text{ and } j \text{ in the same district}] > 0.5$), they are placed in different districts and if $i$ and $j$ were usually separated ( i.e., $\Pr[i \text{ and } j \text{ in the same district}] < 0.5$), they are placed in the same district. A similar interpretation follows for weighted distances, but with the weight $\theta(i,j)$ taken into account, i.e., higher cost if $\theta(i,j)$ is larger.


\subsection{Comparison to Previous Work on Detecting Gerrymandering}\label{app:comparison_Abri}
With the exception of the work of \citet{abrishami2020geometry}, the methods that were introduced such as \citet{herschlag2020quantifying,Chikina2017,ko2022all,lin2022auditing,deford2019recombination} to detect gerrymandering have all been based on election outcomes and therefore are orthogonal to our work. Now, we provide more details about \citet{abrishami2020geometry}. The work of \citet{abrishami2020geometry} introduces a distance measure over redistricting, but finding the distance between two maps requires solving a linear program which is significantly more computationally expensive then our distance measure. Further, \citet{abrishami2020geometry} did not define a medoid or centroid map. Moreover, their work only considered the 2011 and 2016 gerrymandered maps of North Carolina unlike our work which also considered the 2011 gerrymandered instance of Pennsylvania. More importantly, the method used in \citet{abrishami2020geometry} to show that the 2011 and 2016 North Carolina maps are gerrymandered amounts to sampling 100 maps and then applying multidimensional scaling (MDS) a classical method for visualizing high dimensional data in lower dimensional space \cite{borg2013mds,mead1992review}. Figure \ref{fig:Abrishami_fig} is from \cite{abrishami2020geometry} and shows the visualization. Note that the 2012 map in the figure is actually the same as the 2011 map we use (the map was enacted in 2011 but used in 2012). One can see that the gerrymandered maps in red lie further away unlike the judges' map. Based on the above discussion one can identify the following limitations in \cite{abrishami2020geometry}: 
\begin{enumerate}
    \item \textbf{Small Sample Size:} In previous work in redistricting such as \cite{Chikina2017,deford2019recombination,herschlag2020quantifying} we have seen tens and even hundreds of thousands of maps used to estimate a quantity. However, \cite{abrishami2020geometry} only uses 100 maps which is a significant limitation that could put into question the statistical accuracy.  This is an outcome of their distance measure which is computationally expensive.  
    \item \textbf{Heuristic Method with No Outlier Score:} \cite{abrishami2020geometry} uses a high dimensional visualization method (MDS). One necessarily loses information when projecting to a lower dimensional space and the final visualization is not considered to give a necessarily accurate view of the data in higher dimensional space. Furthermore, no outlier score is given for the gerrymandered maps, but rather the observation that they appear further away. Not having a clear numerical value is a clear deficiency in this method. Moreover, in Figure \ref{fig:Abrishami_fig} there exists a large number of ordinary maps sampled from $\recom$ that are comparatively as far away at the  exterior of the visualization as the red gerrymandered maps of North Carolina. 
\end{enumerate}

Our framework resolves both issues above. It can use as much as $200{,}000$ maps or more and gives an outlier score in terms of percentile distance. 

\begin{figure}[h!]
  \centering
  \includegraphics[scale=0.2]{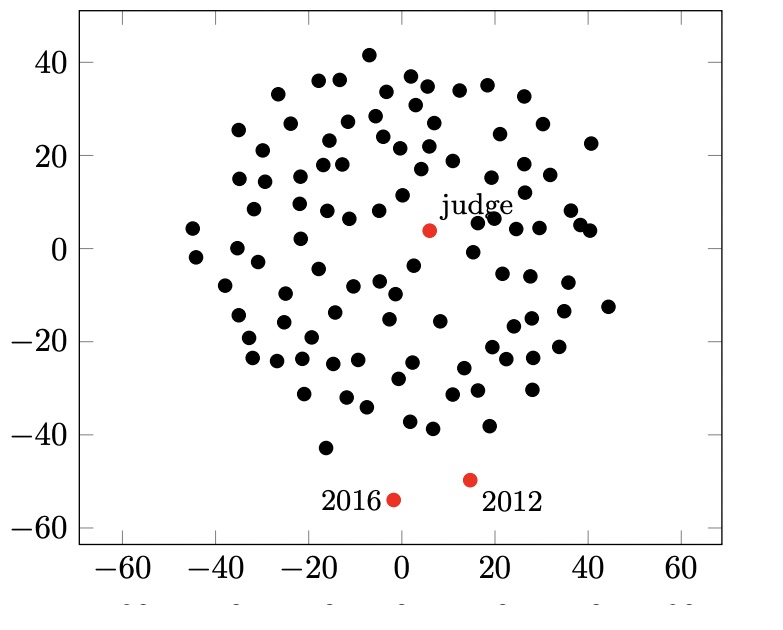}
  \caption{Figure from \cite{abrishami2020geometry} shown as Figure 15 (b) in the reference. The black points are sampled using $\recom$.}
  \label{fig:Abrishami_fig}
\end{figure}

\paragraph{Limitation of Our Method:} Our method is a distance (dissimilarity) based method. As we have seen, it resolves significant issues and introduces a rigorous and statistically stable framework to detect maps that are outlier in terms of distance. However, this may or may not be an advantage since a gerrymandering legislature may select a map that unfairly favors it that is close to the ensemble in terms of distance \cite{duchin2021political}. In general, there is no clean cut approach to detect gerrymandering in any instance. In fact, even using election outcomes it is not in general possible to detect a gerrymandered map. For example, consider the case where the election outcomes are uniformly distributed, this would imply that all outcomes are equally rare.